\documentclass[12pt, draftclsnofoot, onecolumn]{IEEEtran}
\usepackage{graphicx}
\usepackage{amssymb}
\usepackage{amsbsy}
\usepackage{amsmath}
\usepackage{amsthm} 
\usepackage{comment}
\usepackage{bbm}
\usepackage{cite}
\usepackage{caption}
\usepackage{subcaption}
\usepackage{color}
\newtheorem{thm}{Theorem}
\newtheorem{cor}{Corollary}
\newtheorem{lem}{Lemma}

\ifCLASSINFOpdf
\else

\fi

\begin{document}
\title{Treating Interference as Noise in Cellular Networks: A Stochastic Geometry Approach}
\author{\IEEEauthorblockN{Mudasar Bacha, Marco Di Renzo and Bruno Clerckx}\thanks{M. Bacha and B. Clerckx is with the Communication and Signal Processing Group, Department
of Electrical and Electronic Engineering, Imperial College London, London SW7 2AZ, U.K. (e-mail: m.bacha13@imperial.ac.uk; b.clerckx@imperial.ac.uk). }
\thanks{Marco Di Renzo is with the Laboratoire des Signaux et Systemes, CNRS, CentraleSupelec, Univ Paris-Sud, Universite Paris-Saclay, 91192 Gif-sur-Yvette Cedex, France (e-mail: marco.direnzo@l2s.centralesupelec.fr).}
\thanks{This work is partially supported by the UK EPSRC under grant EP/N015312/1.}
}
%
\maketitle
\begin{abstract}
The interference management technique that treats interference as noise (TIN) is optimal when the interference is sufficiently low. Scheduling algorithms based on the TIN optimality condition have recently been proposed, e.g., for application to device-to-device communications. TIN, however, has never been applied to cellular networks. In this work, we propose a scheduling algorithm for application to cellular networks that is based on the TIN optimality condition. In the proposed scheduling algorithm, each base station (BS) first randomly selects a user equipment (UE) in its coverage region, and then checks the TIN optimality conditions. If the latter conditions are not fulfilled, the BS is turned off. In order to assess the performance of TIN applied to cellular networks, we introduce an analytical framework with the aid of stochastic geometry theory. We develop, in particular, tractable expressions of the signal-to-interference-and-noise ratio (SINR) coverage probability and average rate of cellular networks. In addition, we carry out asymptotic analysis to find the optimal system parameters that maximize the SINR coverage probability. By using the optimized system parameters, it is shown that TIN applied to cellular networks yields significant gains in terms of SINR coverage probability and average rate. Specifically, the numerical results show that average rate gains of the order of $21\%$ over conventional scheduling algorithms are obtained.
\end{abstract}
\begin{IEEEkeywords}
Treating interference as noise (TIN), cellular networks, user scheduling, Poisson point processes, stochastic geometry, system level analysis.
\end{IEEEkeywords}
\section{Introduction}
Interference has always been one of the main limiting factors in cellular networks due to its indeterministic nature. 
In order to cope with interference, different solutions have been proposed. For example, coordinated multipoint (CoMP) \cite{CoMP_magzine} and intercell interference coordination (ICIC) \cite{eICIC_ABS} have been considered for the long term evolution-advanced (LTE-A) communications standards. In CoMP operation, multiple base stations (BSs) cooperate over a backhaul link and jointly transmit data to the cell-edge user equipments (UEs) in order to mitigate the intercell interference and hence improve the network throughput \cite{CoMP_IEEE_proc_comn}. The ICIC blanking of macrocells has been proposed for application to heterogeneous networks in order  to reduce the amount of interference to the UEs of small cells \cite{on_off_muting}. Fractional power control (FPC) is considered to be an  essential feature in the uplink (UL) of the  LTE and LTE-A communication standards \cite{LTE_FPC, FPC_VTC}. FPC, however, still generates high levels of interference and limits the UL performance. Another approach is interference aware fractional power control (IAFPC), which limits the maximum amount of interference that a UE can generate. If the interference generated is greater than a maximum threshold, the UE adjusts its transmit power so as to reduce the interference to a given maximum value \cite{FPC_2}.
These few examples confirm the relevance that mitigating interference in cellular networks plays, as well as the research efforts that have been put towards this end. All these approaches, however, are heuristic in nature and do not rely upon any information theoretic optimality conditions.

Treating interference as noise (TIN) is a known interference management technique that is optimal if the strength of the intended link is greater than or equal to the interference strength that the intended link receives multiplied by the interference strength that the intended link creates \cite{optimality_TIN}. It is not just the optimality of TIN that makes it attractive to the research community, but also its low complexity and robustness to channel uncertainty \cite{optimality_TIN,TIN_optimality1,TIN_optimality2,TIN_simplicity1}. Despite its optimality and simplicity, TIN has never been applied to cellular networks. In the present paper, we investigate the application of TIN to cellular networks and introduce a tractable analytical framework in order to optimize its parameters and quantify its achievable performance. To this end, the mathematical tools of stochastic geometry and Poisson point processes are employed.
\subsection{Related Work}
A spectrum sharing mechanism for application in device-to-device (D2D) communications, which is referred to as ITLinQ, has been proposed in \cite{Itlinq}. ITLinQ is based on the TIN optimality condition and it has been shown to provide performance gains compared with other heuristic spectrum sharing algorithms. The authors of \cite{Itlinq} provide a distributed version of ITLinQ, which guarantees some fairness among the links. ITLinQ has further been improved in \cite{itlinq_plus}, where ITLinQ+ has been introduced. By using stochastic geometry, a semi-analytical framework for analyzing ITLinQ has been introduced in \cite{SG_itlinq}. Therein, some adjustable parameters are considered, which can be optimized to get high gain over other D2D scheduling algorithms.

An analytical framework to study the performance of UL heterogeneous cellular networks that employ IAFPC has been proposed in \cite{marco_interference_aware}. According to the IAFPC scheme, the UEs that generate higher interference than a maximum threshold limit their transmit power so that the interference is below an admissible maximum value. A similar approach that turns off the UEs that generate more interference than a maximum threshold has been studied in \cite{marco_muting}. In both \cite{marco_interference_aware} and \cite{marco_muting}, stochastic geometry tools have been used for performance analysis and optimization. Gains in terms of average rate and power consumption have been shown.
Another stochastic geometry based framework that studies the problem of BS cooperation in the downlink for heterogeneous networks can be found in \cite{coordinated_joint_trans}. The analysis of BS cooperation with retransmissions can be found in \cite{spatiotemporal_cooperation}.
Simulations based studies that consider interference aware power control can be found in \cite{simulation_pc1,simulation_pc2,marco_UL}. 

The existing works that employ the TIN optimality condition for interference management are limited to D2D communications \cite{Itlinq,itlinq_plus,SG_itlinq}, whereas the works that employ interference awareness methods in cellular networks are based on heuristic criteria \cite{marco_interference_aware, FPC_2, marco_muting,coordinated_joint_trans,spatiotemporal_cooperation, simulation_pc1,simulation_pc2,marco_UL}. In the present paper, we propose a scheduling algorithm based on the TIN optimality condition and develop a tractable analytical framework for its analysis and optimization by using tools from stochastic geometry. In the last few years, stochastic geometry has emerged as a power tool for the analysis of wireless networks due to its analytical tractability yet accuracy. For example, it has been used for the analysis of cellular networks \cite{dude_mudasar,Geff_DL,Geff_K_tier_DL,MGF_marco,martin_RatelessCodes}, cognitive networks \cite{SG_TWC_martin},  millimeter wave cellular networks \cite{marco_milimeter,heath_millimeter_wave}, ad-hoc networks \cite{geff_adhoc}, wireless powered cellular networks \cite{marco_wpt,marco_EE}, and backscatter communication networks \cite{backscatter_mudasar}.











\subsection{Contributions and Outcomes}
Designing a cellular network based on the TIN optimality conditions and developing an analytical framework for its analysis and optimization is a challenging task. To this end, in fact, a centralized controller that keeps track of all the channels of the cellular network and that identifies the strongest interference received and the strongest interference created on each link of the cellular network is needed. The complexity of such a centralized scheduler may be too high for application to large-scale networks. The resulting scheduling algorithm, in addition, would be intractable from the analytical standpoint and, therefore, difficult to optimize without using extensive system-level simulations.

To deal with the issues of implementation complexity and analytical intractability, we propose and study the performance of a simplified two-step version of the (optimal or centralized) TIN-based scheduling algorithm, which can be implemented in a distributed manner and requires only the channel state information (CSI) of neighboring BSs. In the first step, each BS randomly selects a UE in its coverage region, similar to conventional cellular networks that do not employ TIN. In the second step, only the BSs that fulfill the TIN optimality conditions schedule for transmission the UE in their coverage region. The rest of the BSs, on the other hand, are turned off. In order to make the TIN scheduling algorithm suitable for application to cellular networks, we introduce two design parameters ($M$ and $\mu$) that are optimized in order to control the number of BSs that are turned off. The optimization of these two parameters is important in order to identify a suitable trade-off between the potential reduction of interference and the potential loss of average rate that turning some BSs off entails. In spite of the latter potential loss of average rate, our analysis shows that the proposed TIN-based scheduling algorithm outperforms conventional cellular networks in terms of average rate, thanks to its effective interference reduction capability.

The proposed two-step TIN-based scheduling algorithm is simple to implement in cellular networks. Its analysis and optimization are, however, still challenging. The main reason is the lack of analytical results for the distribution of the downlink distances within the typical cell of a Voronoi tessellation \cite{martin_RatelessCodes,marco_EE}. To overcome this limitation, we introduce some approximations that lead to a tractable analytical framework, which is shown to be suitable for system optimization yet sufficiently accurate. We approximate, in particular, the typical cell of a Voronoi tessellation with the so-called Crofton cell \cite{marco_EE}, and the distribution of the distance between a BS and its most interfered UE with the distribution of the distance between the typical UE and the BS that creates the highest interference to it. These two approximations result in a lower bound for the coverage probability of the proposed two-step TIN-based scheduling algorithm.



The unique contributions and outcomes of the present work can be summarized as follows:
\begin{itemize}
\item{} We derive tractable analytical expressions for the SINR coverage probability and average rate of the proposed two-step (distributed) TIN scheduling protocol. For specific system parameters, the analytical formulas are proved to reduce to the SINR coverage probability and average rate of conventional cellular networks.
\item{} We show that  unique values of $M$ and $\mu$ that maximize the SINR coverage probability and average rate exist. In order to compute such optimal $M$ and $\mu$, we carry out asymptotic analysis of the SINR coverage probability and provide a simple optimization algorithm. By setting $M=1$, more precisely, we identify the optimal value of $\mu$ for high and low values of the SINR decoding threshold.
\item{} We observe that the optimal value of $\mu$ decreases as the SINR threshold increases. This implies that, in order to achieve a high SINR decoding threshold, more BSs need to be turned off. If, on the other hand, the SINR threshold is small, we show that no BSs need to be turned off, which implies that optimal performance can be obtained by using the conventional scheduling algorithm.
\item{} We further show that TIN-based scheduling algorithms with optimized parameters significantly improve the SINR coverage probability and average rate.
For example, the exact  implementation (through simulation without any approximation)  of the two-step TIN-based scheduling  improves the SINR coverage probability by $67\%$ and the mathematically tractable implementation (lower bound)  of the two-step TIN-based scheduling improves the SINR coverage probability by $36\%$, if the SINR decoding threshold is set to $10$ dB.
In addition, the corresponding increase of the average rate is $21\%$ and $11\%$, respectively, despite the fact that some BSs are turned off, compared to the classical scheduling algorithm.
\end{itemize}

The rest of the present paper is organized as follows. In Section \ref{system_model}, we introduce the network model and the TIN-based scheduling algorithm. In Section \ref{SINR_SE_analysis}, we provide a tractable analytical framework of the SINR coverage probability and average rate. In Section \ref{asymptotic_analysis}, we carry out asymptotic analysis in order to find the optimal system parameters that optimize the SINR coverage probability. In Section \ref{results}, simulation and numerical results are presented. Finally,  Section \ref{conclusion} concludes the paper and summarizes some ideas for future research.
\section{System Model}
\label{system_model}
\subsection{Network Model}
We consider a single-tier downlink cellular network in which the locations of BSs and UEs are modeled as points of two bi-dimensional and mutually independent homogeneous Poisson point processes (PPPs). We denote by $\Phi_b$ and $\Phi_u$ the PPPs, and by $\lambda_b$ and $\lambda_u$ the densities of BSs and UEs, respectively. The density of UEs is assumed to be much greater than the density of BSs. Thus, all the BSs are active and have UEs to serve in every resource block (carrier frequency, time slot, etc.), if no scheduling for interference management is applied. We assume full frequency reuse, i.e., all the BSs share the same transmission bandwidth. Each UE is associated with the nearest BS. Accordingly, the coverage regions of the BSs constitute a Poisson-Voronoi tessellation in the plane. The BSs and UEs have a single antenna. The standard unbounded path-loss model with path-loss exponent $\alpha>2$ is considered. The fast-fading is assumed to follow a Rayleigh distribution. More general system models may be analyzed. In the present work, however, we consider the so-called standard modeling assumptions \cite{marco_EE}, in order to focus our attention on the impact and potential benefits of TIN.
\subsection{Treating Interference as Noise}
Treating interference as noise is a scheduling algorithm that has attracted major interest for practical and theoretical reasons. From the implementation point of view, TIN is attractive due to its simplicity and adaptability to channel uncertainties. From the theoretical standpoint, it is a promising solution due to its optimality under certain conditions \cite{optimality_TIN}.

The TIN optimality conditions, more precisely, can be stated as follows: In a wireless network with $n$ transmitter and receiver pairs, if the strength of the intended link (from an intended BS to an intended UE) is greater than or equal to the product of the strengths of the strongest interference that the intended BS creates and of the strongest interference that the intended UE receives, then TIN achieves the whole capacity region to within a constant gap of $\log3n$ \cite{optimality_TIN}. In mathematical terms, the TIN optimality conditions can be formulated as follows:
\begin{equation}
\mathrm{{{SNR}}}_{i}\geqslant\max_{j\neq i}\mathrm{\mathrm{{{INR}}}}_{ij}\max_{k\neq i}\mathrm{\mathrm{{{INR}}}}_{ki}\quad\forall i=1,...,n,
\label{TIN_original}
\end{equation}
where $\mathrm{{{SNR}}}_{i}$  and $\mathrm{{{INR}}}_{ij}$ denote the signal-to-noise ratio (SNR) of link $i$ and the interference-to-noise ratio (INR) of source $j$ at destination $i$, respectively. It is worth mentioning, in particular, that the SNRs and INRs in \eqref{TIN_original} depend on both the path-loss and the fast fading.
\subsection{TIN-Based Scheduling Algorithm}
The scheduling algorithm based on the TIN optimality conditions in \eqref{TIN_original} may be difficult to realize in large-scale cellular networks, since a centralized controller that is aware of the (instantaneous) channel state information of all the links available in the networks is necessary. In the present paper, this approach is referred to as \textit{centralized TIN scheduling}.

To overcome this issue, we propose a two-step distributed scheduling algorithm inspired by \eqref{TIN_original}, which, in addition, does not necessitate the instantaneous channel state information of the links. This approach is referred to as \textit{two-step or distributed TIN scheduling}. In the first step, for each available resource block, each BS randomly selects a UE that lies in its coverage region. In the second step, the BSs that do not fulfill the following (simplified) TIN optimality conditions:
\begin{equation}
M\mathrm{{\overline{SNR}}}_{i}^{\mu}\geqslant\max_{j\neq i}\mathrm{\mathrm{{\overline{INR}}}}_{ij}\max_{k\neq i}\mathrm{\mathrm{{\overline{INR}}}}_{ki}\quad\forall i=1,...,n,
\label{TIN_basic}
\end{equation}
are turned off, where $\mathrm{{{\overline{SNR}}}}_{i}$  and $\mathrm{{{\overline{INR}}}}_{ij}$ have the same meaning as in \eqref{TIN_original} except that they are averaged with respect to the fast fading in order to dispense the scheduler from requiring the instantaneous channel state information of the links. Equation \eqref{TIN_basic}, as opposed to \eqref{TIN_original}, is applied by each BS independently of the other BSs, which makes it suitable for a distributed implementation and requires only the average channel state information of neighboring BSs. In \eqref{TIN_basic}, in addition, we have introduced two design parameters, $M$ and $\mu$, for system optimization, which allow us to control the number of BSs that are turned off, and, thus, are instrumental for interference management. In particular, $M$ is a positive real number greater than or equal to one ($M \ge 1$), and $\mu$ is a positive real number greater than or equal to one and less than or equal to two ($1 \le \mu \le 2$).

In order to get deeper understanding and insight from \eqref{TIN_basic}, we rewrite it with the aid of a more explicit notation by taking as an example the cellular network realization depicted in Fig. \ref{TIN_diagram}. For any realization of the cellular network under analysis, we select one cell (BS) at random that is referred to as \textit{the typical cell or the typical BS}. Among the UEs that lay in its coverage region, we select one UE at random for every available resource block. We focus our attention on a randomly chosen resource block and the UE that is scheduled for transmission on it is referred to as the \textit{typical UE}. Let $X_{11}$ denote the downlink distance between the typical UE (UE1 in Fig. \ref{TIN_diagram}) and the typical BS (BS1 in Fig. \ref{TIN_diagram}), $X_{12}$ denote the downlink distance between the typical BS and the most interfered UE (UE2 in Fig. \ref{TIN_diagram}) by it, and $X_{21}$ denote the downlink distance between the typical UE and the BS (BS2 in Fig. \ref{TIN_diagram}) that creates the highest interference to it. With the aid of this explicit notation, \eqref{TIN_basic} can be re-formulated as follows:
\begin{equation}
M\left(\frac{PX_{11}^{-\alpha}}{N}\right)^{\mu}\geq\frac{PX_{12}^{-\alpha}}{N}\frac{PX_{21}^{-\alpha}}{N}
\label{TIN_cell_basic},
\end{equation}
where $P$ is the transmit power of the BSs, $N$ is the noise power at the UE, and $\alpha$ is the path-loss exponent. From \eqref{TIN_cell_basic}, it follows that the typical BS necessitates only three average SNRs to check the TIN optimality conditions, and, thus, a network-level controller is not necessary.

Equivalently, \eqref{TIN_cell_basic} can be written as follows:
\begin{equation}
X_{11}\leq M^{\frac{1}{\alpha\mu}}\left(\frac{N}{P}\right)^{\frac{2-\mu}{\alpha\mu}}\left(X_{12}X_{21}\right)^{\frac{1}{\mu}}
\label{TIN_cell}.
\end{equation}
\begin{figure}
	\centering
		\includegraphics[scale=0.7]{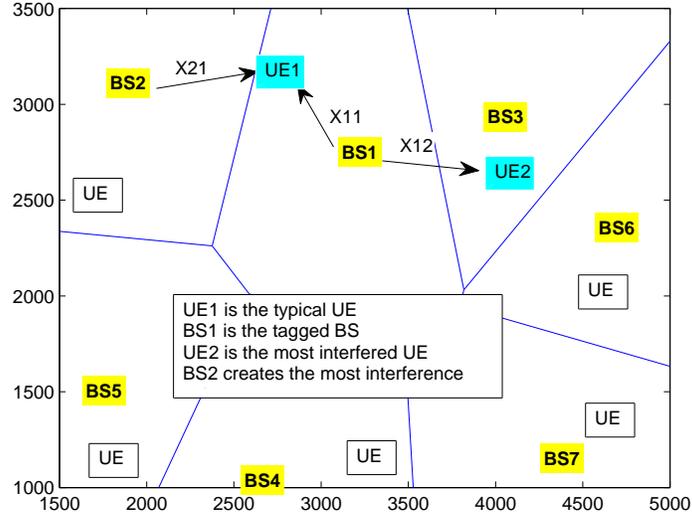}
	\caption{Network model.}
	\label{TIN_diagram}
\end{figure}
During the second step of the proposed TIN-based scheduling algorithm, the typical UE is scheduled for transmission if $X_{11}$ fulfills the constraint in \eqref{TIN_cell}. Otherwise, the typical BS is turned off and the UEs that lay in its coverage region are not scheduled for transmission. This implies that the interference can be potentially reduced at the cost of reducing the average rate and the fairness among the UEs, since some of them are not scheduled for transmission at a given time instance. In the present paper, we employ the pair of parameters $M$ and $\mu$ to find the good balance in order to reduce the interference while increasing both the SINR coverage probability and the average rate. We do not explicitly study, on the other hand, the issue of fairness among the UEs. This is postponed to future research. It is important to mention, however, that the fairness among the UEs may not be a critical issue in a well designed system. If in a given time instance, e.g., a time-slot, the TIN optimality condition in \eqref{TIN_cell} is not fulfilled, it may be likely fulfilled in another (the next) time instance. The local interference conditions that are considered in \eqref{TIN_cell} may, in fact, change at every time instance for two main reasons. i) In each resource block, the BSs select the UEs at random. Accordingly, the interference perceived by the typical UE may change every time that a random UE is chosen by the BSs. ii) Mobility and shadowing, which are not explicitly considered in the present work, may change the interference perceived by the typical UE as well. The system parameters $M$ and $\mu$, in addition, may be tuned in order to account for specific fairness requirements. By setting, for example, $M$ to a very large value, the scheduling criterion in \eqref{TIN_cell} is deactivated and the typical fairness requirements among the UEs can be guaranteed. In the sequel, it is shown, more precisely, that our system model reduces to the original cellular network model without TIN-based scheduling by setting $M=1$ and $\mu=2$. Due to space limitations, the analysis and optimization of $M$ and $\mu$ in order to find a good trade-off between average rate and user fairness is postponed to future research.
%
\begin{figure*}
\centering
			\begin{subfigure}{0.45\textwidth}
			\centering
			\includegraphics[width=1\linewidth]{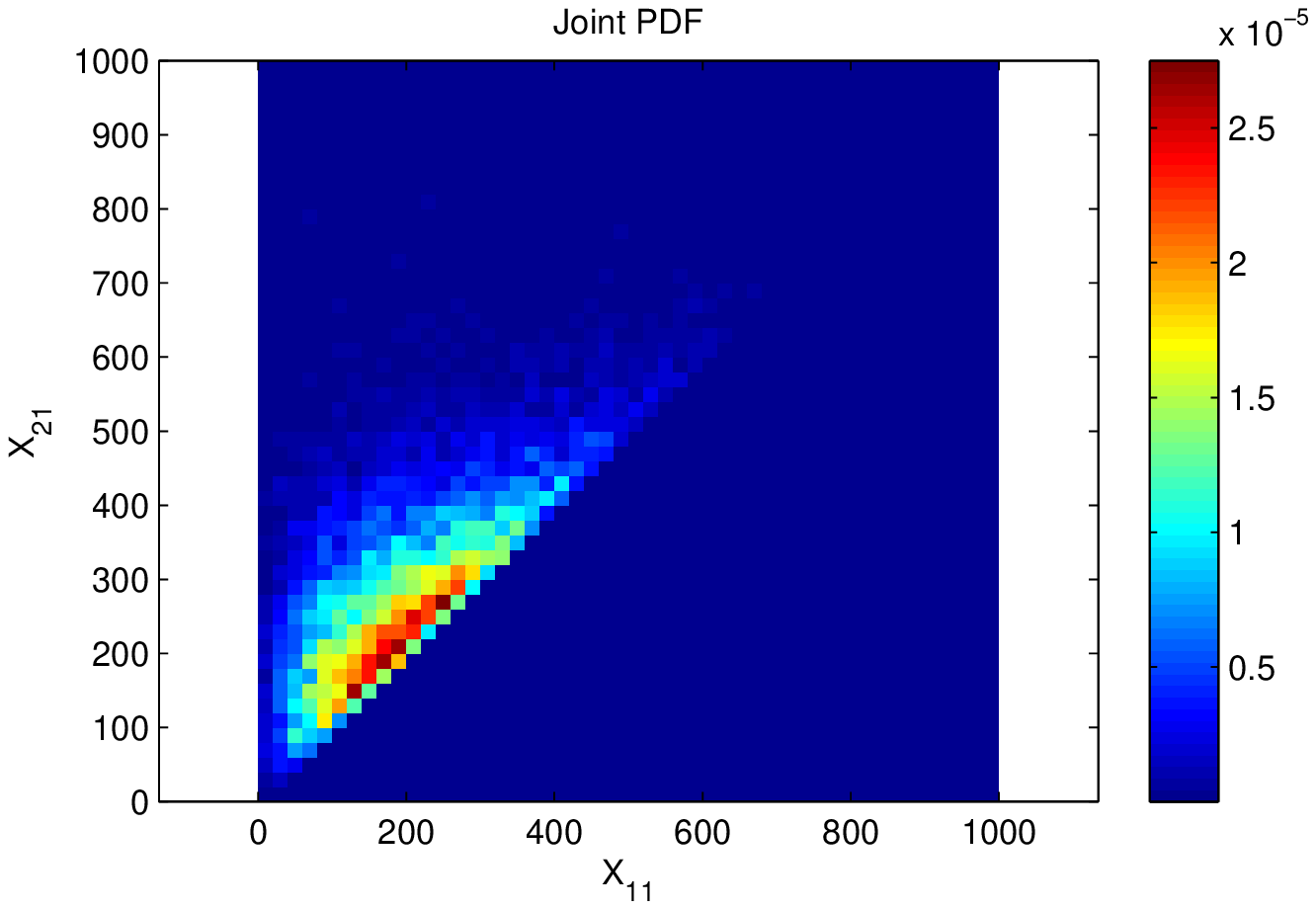}
			\caption{Conventional cellular network.}
			\label{classical_x11_x21}
		\end{subfigure}
	\begin{subfigure}{0.45\textwidth}
		\centering
		\includegraphics[width=1\linewidth]{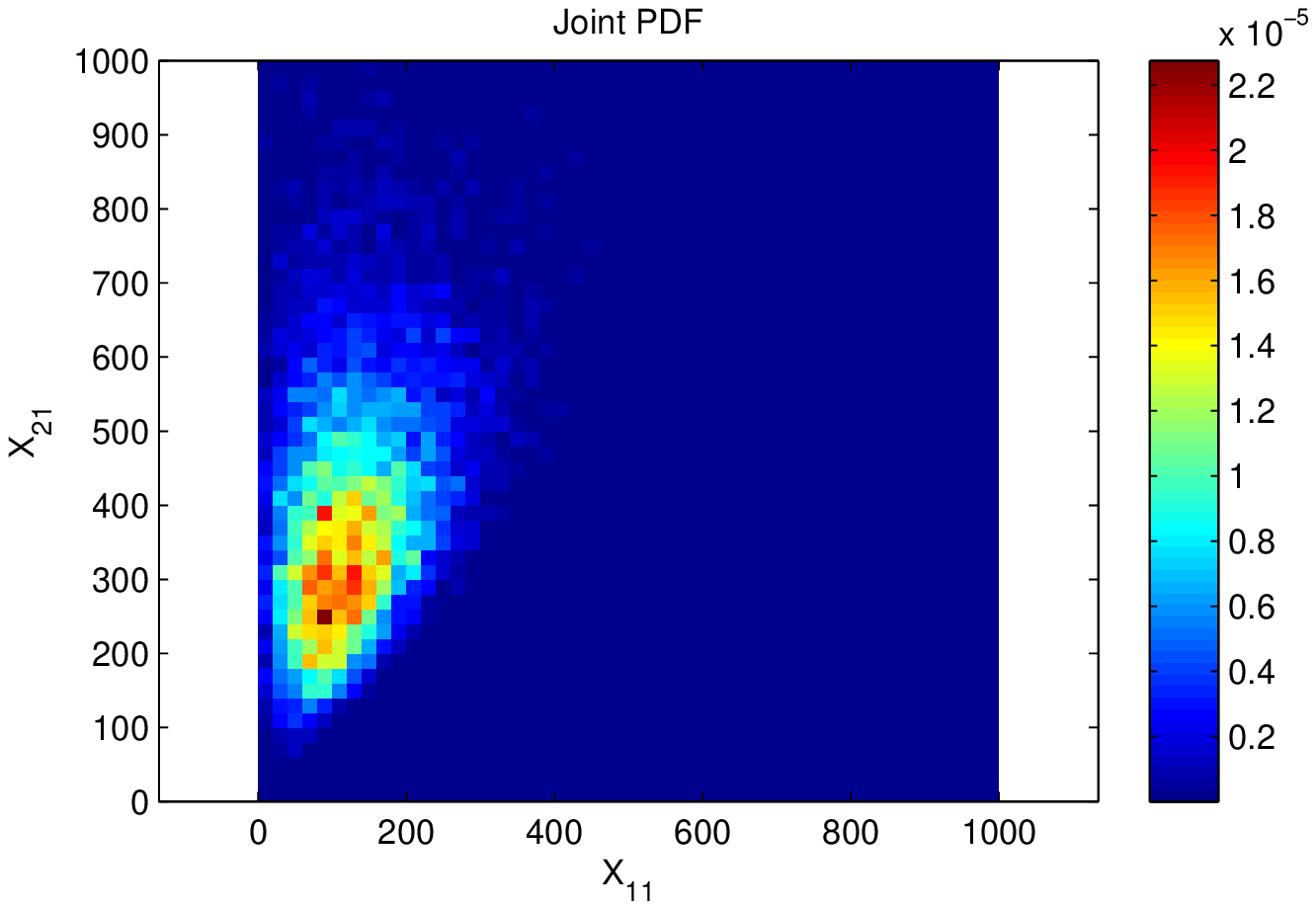}
		 \caption{TIN-based cellular network.}
		 \label{TIN_x11_x21}
	\end{subfigure}\hfill
\centering
			\begin{subfigure}{0.45\textwidth}
			\centering
			\includegraphics[width=1\linewidth]{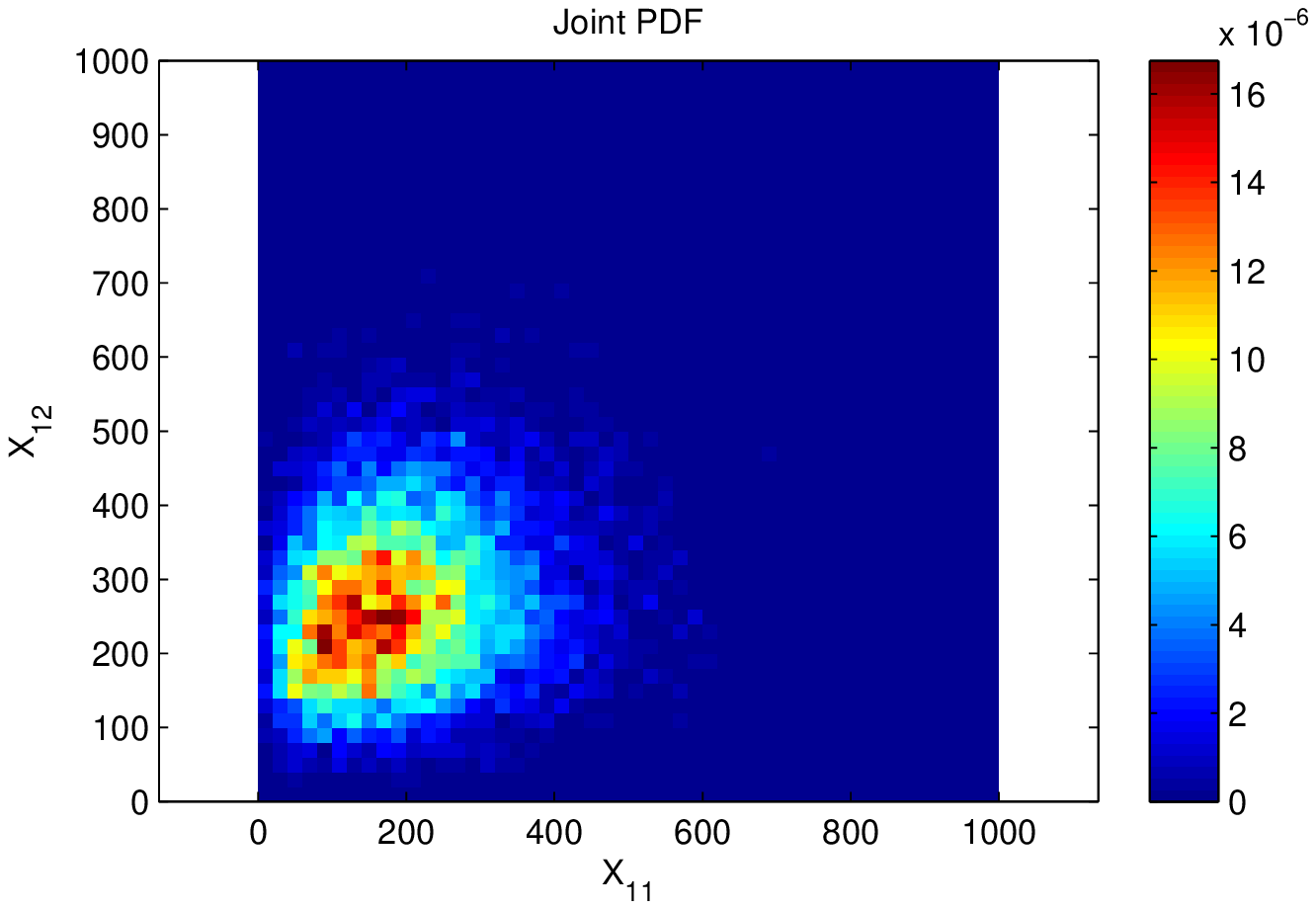}
			\caption{Conventional cellular network.}
			\label{classical_x11_x12}
		\end{subfigure}
	\begin{subfigure}{0.45\textwidth}
		\centering
		\includegraphics[width=1\linewidth]{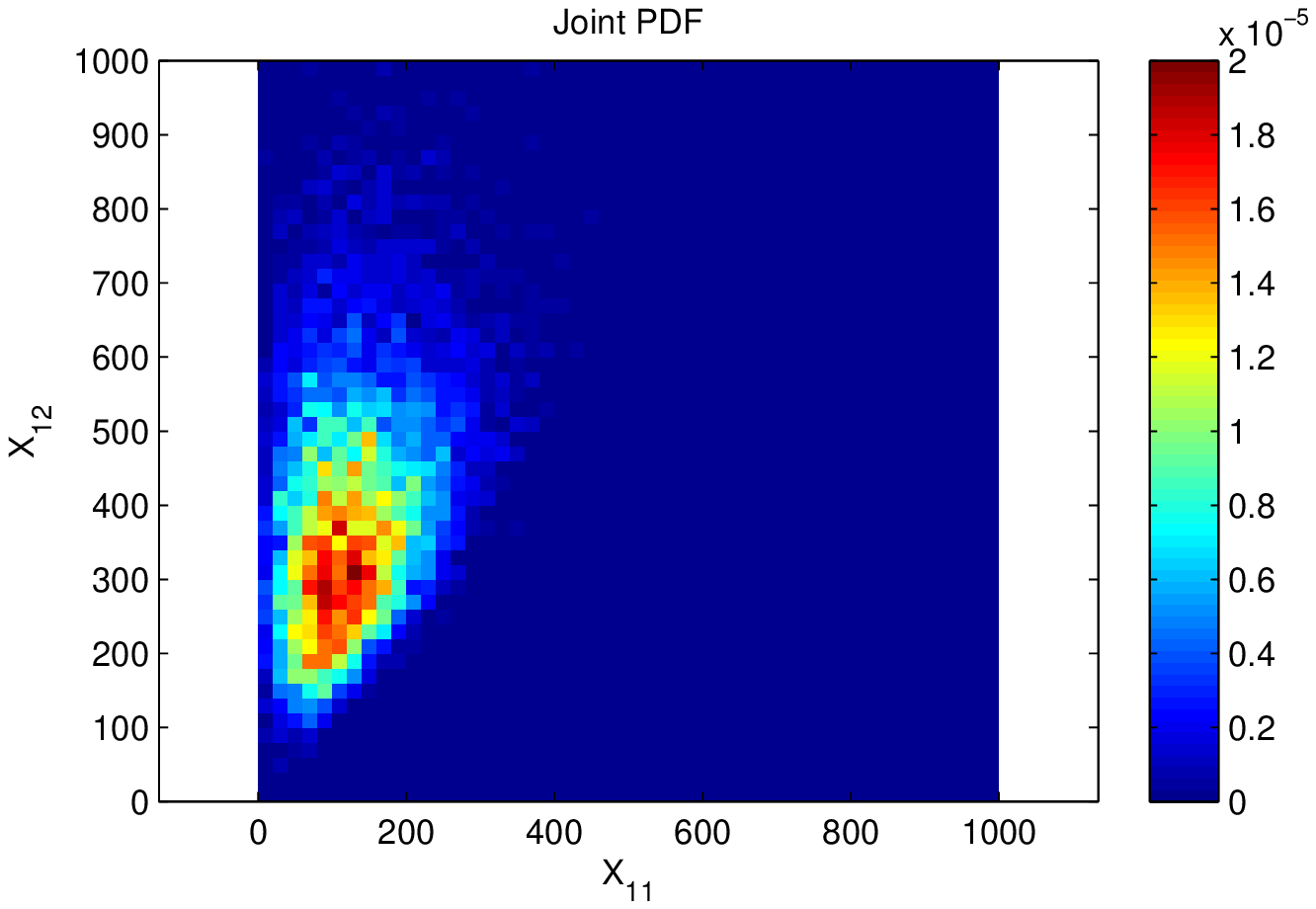}
		 \caption{TIN-based cellular network.}
		 \label{TIN_x11_x12}
	\end{subfigure}\hfill
\caption{Top view of the joint probability density functions of $(X_{11}, X_{21})$ and $(X_{11}, X_{12})$ in conventional and TIN-based cellular networks ($\lambda_b=5$, $M=1$, $\mu=1.8$).}
		\label{empirical_joint_pdfs}
\end{figure*}


To better understand the impact of the proposed TIN-based scheduling algorithm on the downlink distances in \eqref{TIN_cell}, we analyze the joint probability density function of the distance pairs $(X_{11}, X_{21})$ and $(X_{11}, X_{12})$. We compare, in particular, a cellular network where TIN is not applied (i.e., each BS selects a UE at random in its coverage region, which is always scheduled for transmission) and the same cellular network where the scheduling criterion in \eqref{TIN_cell} is applied. From Fig. \ref{classical_x11_x21}, if TIN is not applied, we evince that $X_{21}$ is always greater than $X_{11}$, but the most interfering BS may be located just outside the coverage region of the typical BS. If TIN is applied, on the other hand, Fig. \ref{TIN_x11_x21} shows that, as opposed to Fig. \ref{classical_x11_x21}, the most interfering BS is located further away from the coverage boundary of the typical BS. This confirms that TIN is capable of reducing the interference at the typical UE. From Fig. \ref{classical_x11_x12}, if TIN is not applied, we evince that $X_{12}$ may be greater or less than $X_{11}$. This implies that the most interfered UE may be located farther or closer than the typical UE. If TIN is applied, on the other hand, Fig. \ref{TIN_x11_x12} shows that the most interfered UE is located farther than the typical UE. This highlights that TIN is capable of reducing the interference towards the UEs.

In summary, the TIN-based scheduling algorithm is capable of reducing the interference in cellular networks by turning off those BSs that create a high level of interference, as well as those BSs whose tagged UEs receive a high level of interference. The proposed TIN-based scheduling algorithm, in particular, is different from those reported in \cite{marco_interference_aware} and \cite{marco_muting}, where the rationale is to compare the signal strengths against a maximum but fixed level of tolerable interference. In the proposed TIN-based scheduling algorithm, the level of tolerable interference depends on the signal strengths themselves, which makes TIN robust to channel uncertainty as well. Also, the proposed TIN-based scheduling not only accounts for the amount of interference that is received but also for the amount of interference that is generated.

\section{SINR Coverage Probability and Average Rate}
\label{SINR_SE_analysis}
In this section, we introduce analytical frameworks for computing the SINR coverage probability and average rate of TIN-based cellular networks. The obtained analytical frameworks are instrumental to quantify the performance gains offered by TIN, to compare conventional against TIN-based cellular networks, as well as to optimize the system parameters in order to identify the correct tradeoff between interference reduction and the required average rate.

To this end, some enabling results are needed. In particular, the probability that a random UE satisfies the TIN optimality conditions in \eqref{TIN_cell}, and the joint and marginal distributions of the distances $X_{11}$, $X_{12}$, and $X_{21}$ are needed. However, they are not available in the open technical literature. To overcome this analytical challenge, we resort to some approximations that are introduced, motivated, and discussed in the following section for ease of exposition, since they are applied throughout the rest of the present paper.

\subsection{Approximations for Tractable Analytical Modeling} \label{section_approximation}
Three main approximations are used for our analysis, which are detailed as follows.
\begin{itemize}
  \item In our network model, the distribution of the downlink distances within the typical Poisson-Voronoi cell are needed. It is known, however, that these distributions are unknown \cite{martin_RatelessCodes,marco_EE}. A tractable and accurate approximation that is typically employed to overcome this limitation consists of approximating the typical cell with the Crofton (or zero) cell. The Crofton cell of a Poisson-Voronoi tessellation is the cell that contains the origin. It is known that the Crofton cell is larger than the typical cell, but the two cells are equal in law \cite{MDR3}. Some discussions accompanied by empirical results are available in \cite{MDR4}. To obtain a tractable yet accurate analytical framework of the SINR coverage probability and average rate, we conduct the analysis for the Crofton cell instead of for the typical cell. This approach is motivated by the fact that the marginal and joint distributions of $X_{11}$ and $X_{21}$ are available in closed-form for the Crofton cell of a Poisson-Voroni tessellation \cite{distance_distribution}. Since it is known that large cells are more likely to contain the origin, the Crofton cell is larger than the typical cell defined through the Palm measure \cite{MDR6}. This implies that, with the proposed approximation, the distances $X_{11}$ and $X_{21}$ are overestimated.





%
%
%
  \item In our network model, the distribution of the downlink distance, $X_{12}$, between the typical BS and its most interfered UE is needed. To the best of our knowledge, the distribution of this distance is not available in closed-form. In order to overcome this limitation, we still rely on the Crofton cell approximation, and, in addition, we propose to approximate the distribution of $X_{12}$ with the distribution of $X_{21}$, i.e., with the distance between the typical UE and its most interfering BS. This approximation is empirically supported by comparing Fig. \ref{TIN_x11_x21} and Fig. \ref{TIN_x11_x12}, where it is shown that, if TIN is applied, the conditional probability density functions of $X_{12}$ and $X_{21}$ are similar. In Fig. \ref{empirical_cdfs}, in addition, we depict the corresponding cumulative distribution functions of $X_{12}$ and $X_{21}$. We observe that they are not so different from each other, especially for short distances. Furthermore, it is apparent that $X_{21}$ overestimates $X_{12}$.
  \item By applying the TIN-based scheduling algorithm, some BSs are turned off if the TIN optimality conditions in \eqref{TIN_cell_basic} are not fulfilled. Even though the point process of BSs is a PPP, the point process of the active BSs after applying TIN is not a PPP anymore. The TIN-based scheduling algorithm, in fact, introduces some spatial correlations among the set of active BSs that depend on the amount of downlink interference generated and received throughout the entire cellular network. To the best of our knowledge, no exact analytical characterization of the point process of the active BSs exists in the open technical literature. For analytical tractability, and similar to \cite{marco_interference_aware,marco_muting,dude_mudasar}, we approximate the point process of the active BSs with an inhomogeneous PPP. The spatial inhomogeneity is, in particular, determined by the spatial constraints imposed by the TIN optimality conditions in \eqref{TIN_cell_basic}, which allows us to account, at least in part, for the spatial correlations among the active BSs. The details of the approximating inhomogeneous PPP are provided in the sequel.
\end{itemize}

Based on the Crofton cell approximation, the joint probability density function of $X_{11}$ and $X_{21}$ is approximated as follows \cite{distance_distribution}:
\begin{equation}
f_{X_{11},X_{21}}(x_{11},x_{21}) \approx  \left(2\pi\lambda_b\right)^2 x_{11}x_{21} \mathrm{e}^{-\pi\lambda_b x_{21}^2},
\label{joint_pdf_dist}
\end{equation}
if $x_{11} < x_{21}$ and $f_{X_{11},X_{21}}(x_{11},x_{21}) =0$ otherwise. Also, we have $f_{X_{11}}(x_{11}) \approx  2\pi\lambda_b x_{11} \mathrm{e}^{-\pi\lambda_b x_{11}^2}$ and $f_{X_{11},X_{21}}(x_{11},x_{21}) \approx f_{X_{11}}(x_{11}) f_{X_{21}|X_{11}}(x_{21}|x_{11})$, where $f_{X_{21}|X_{11}}(x_{21}|x_{11})$ is the conditional probability density function of $X_{21}$. Based on the second approximation, furthermore, we assume $f_{X_{12}}(x_{12}) \approx f_{X_{21}}(x_{21})$. By capitalizing on the Crofton cell approximation, in addition, the typical UE can be assumed to be at the origin without loss of generality.

Based on the inhomogeneous PPP approximation, the point process of the interfering BSs after applying the TIN optimality conditions is approximated with an inhomogeneous PPP of distance-dependent density $\lambda_{I}(r)$ defined as follows:
\begin{equation}
\begin{cases}
\lambda_{I}(r)= 0 & \quad \text{if} \quad r < \mathcal{R_I}  \\
\lambda_{I}(r)=\lambda_{b}\mathbb{P}\left[\mathbb{A_{UE}}\right] & \quad \text{if} \quad r \ge \mathcal{R_I},
\label{den_interf_BS}
\end{cases}
\end{equation}
where $\mathbb{A_{UE}}$ denotes the event that \eqref{TIN_cell_basic} is true, $\mathbb{P}\left[\mathbb{A_{UE}}\right]$ is its probability of occurrence, and $\mathcal{R_I}$ constitutes the smallest distance from the origin of any interfering BSs after applying TIN. $\mathcal{R_I}$ is referred to as the \textit{inhomogeneity ball} and can be formulated as follows:
\begin{equation}
\mathcal{R_I} = \max(X_{11},X_{11}^{\frac{\mu}{2}}\left(\frac{P}{N}\right)^{\frac{2-\mu}{2\alpha}}\left(\frac{1}{M}\right)^{\frac{1}{2\alpha}}),
 \label{interf_boundary}
\end{equation}
where $\max\left(\cdot,\cdot\right)$ denotes the maximum operator, and its first and second arguments originate from the shortest distance cell association criterion and from the TIN-based optimality conditions in \eqref{TIN_cell_basic}, respectively. From \eqref{interf_boundary}, the condition $X_{21}\geq \mathcal{R_I}$ holds implicitly true.

The accuracy of the proposed approximations is analyzed in Section \ref{results} with the aid of Monte Carlo simulations. It is shown that the proposed approximations lead to a tractable yet accurate analytical formulation of the SINR coverage probability and average rate. In the sequel, the proposed approximations are used for all the analytical derivations.

\begin{figure*}
\centering
			\begin{subfigure}{0.45\textwidth}
			\centering
			\includegraphics[width=1\linewidth]{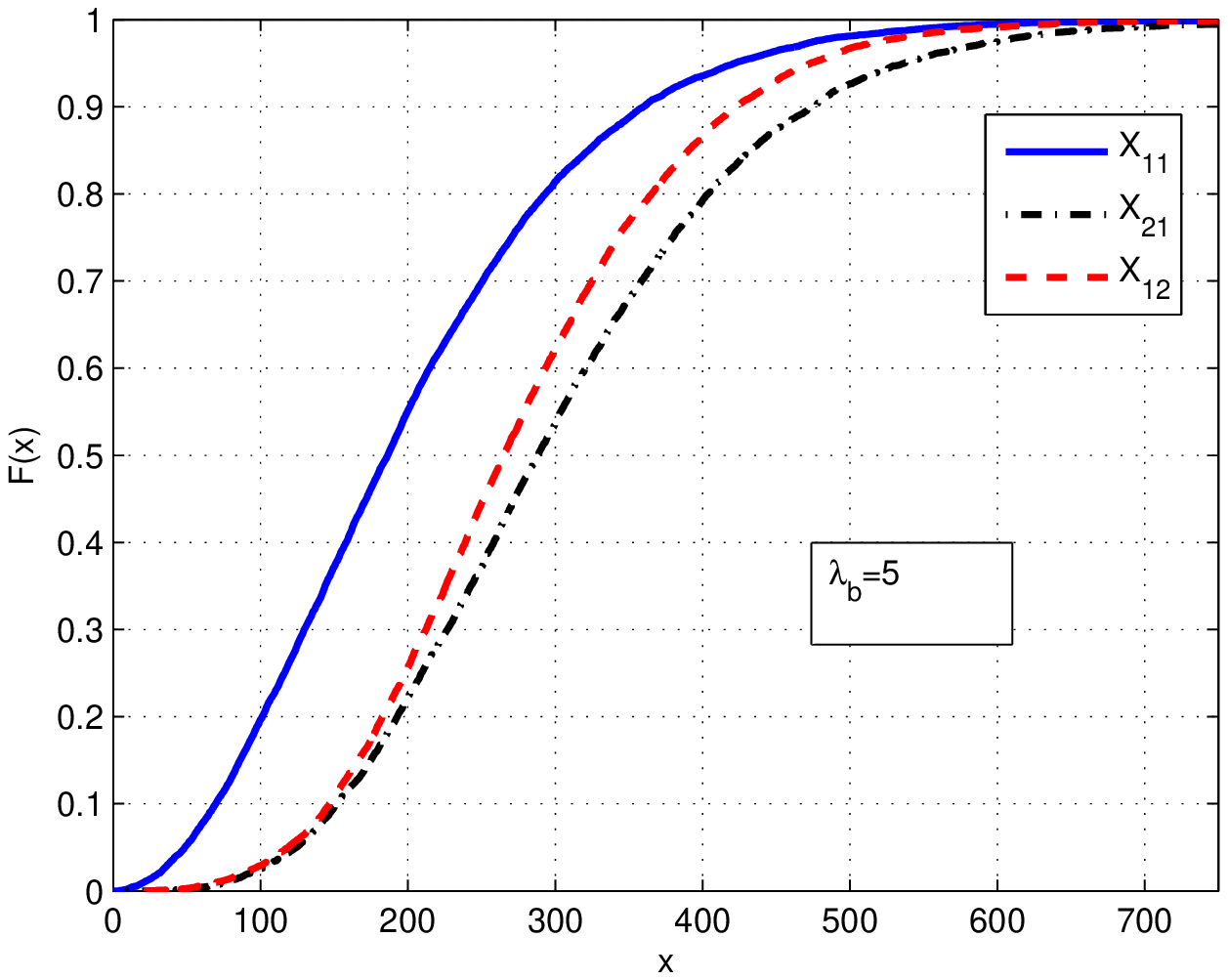}
			\caption{}
			\label{}
		\end{subfigure}
	\begin{subfigure}{0.45\textwidth}
		\centering
		\includegraphics[width=1\linewidth]{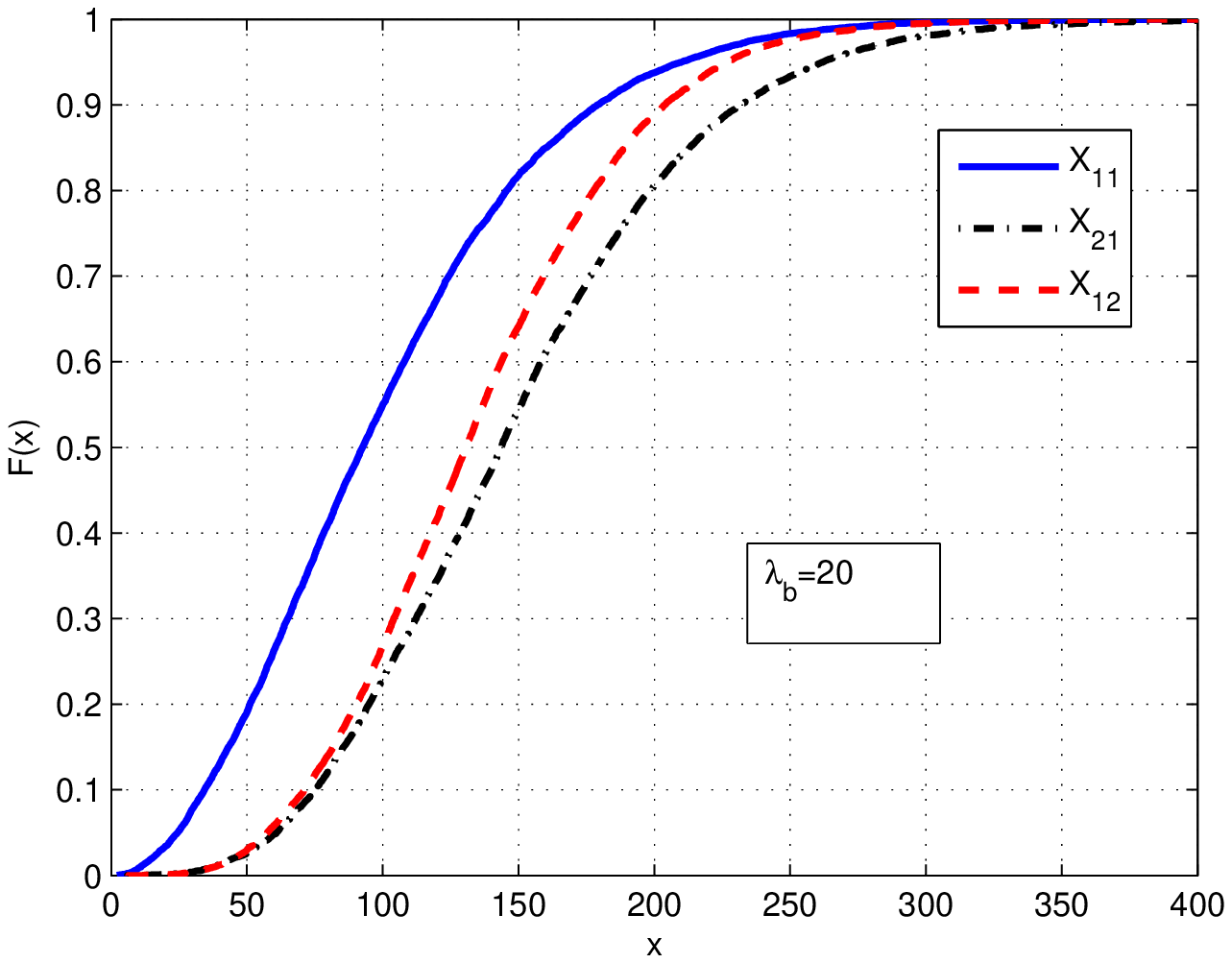}
		 \caption{}
		 \label{}
	\end{subfigure}\hfill
\caption{Empirical cumulative distribution functions of $X_{11}, X_{21}$ and $X_{12}$}
		\label{empirical_cdfs}
\end{figure*}

\subsection{Probability of TIN}
In this section, we compute the probability that a randomly selected UE satisfies the TIN optimality conditions in \eqref{TIN_cell}. Based on the approximations in Section \ref{section_approximation}, the event, $\mathbb{A_{UE}}$, that the typical UE fulfills the TIN optimality conditions can be formulated as follows:
\begin{equation}
\mathbb{A_{UE}}=\left[X_{11}\leq M^{\frac{1}{\alpha\mu}}\left(\frac{N}{P}\right)^{\frac{2-\mu}{\alpha\mu}}X_{21}^{\frac{2}{\mu}}\right]
\label{TIN_cell_received}.
\end{equation}

The probability that the typical UE fulfills the event $\mathbb{A_{UE}}$ is given in the following lemma.
\begin{lem}
The probability that the typical UE is active can be formulated as follows:
\begin{equation}
\begin{split}
\mathbb{P\left[A_{UE}\right]}&=\mathbb{P}\left[X_{11}\leq M^{\frac{1}{\alpha\mu}}\left(\frac{N}{P}\right)^{\frac{2-\mu}{\alpha\mu}}X_{21}^{\frac{2}{\mu}}\right]\\
&= \int_{0}^{\infty}2\left(\pi\lambda_{b}\right)^2 x_{21}\mathrm{e}^{-\pi\lambda_{b}x_{21}^{2}} {\min{} ^2} \left(x_{21},M^{\frac{1}{\alpha\mu}}\left(\frac{N}{P}\right)^{\frac{2-\mu}{\alpha\mu}}x_{21}^{\frac{2}{\mu}}\right)\mathrm{d}x_{21},
\label{TIN_prob_received}
\end{split}
\end{equation}
where $\min\left(\cdot, \cdot\right)$ denotes the minimum function.
\label{prob_TIN_lem}
\end{lem}
\begin{proof}
The probability that a UE is active can be written as follows:
\begin{equation}
\mathbb{P\left[A_{UE}\right]}=\mathbb{E}_{X_{11}, X_{21}}\left[\mathbf{1}\left(X_{11}\leq X_{21}\right)\times\mathbf{1}\left(X_{11}\leq \left( M^{\frac{1}{\alpha\mu}}\left(\frac{N}{P}\right)^{\frac{2-\mu}{\alpha\mu}}X_{21}^{\frac{2}{\mu}}\right)\right)\right], \label{temp2}
\end{equation}
where the fist indicator function, $\mathbf{1}\left(\cdot\right)$, is due to the cell association criterion based on the shortest distance and the second indicator function is due to the TIN optimality conditions. The expectation in \eqref{temp2} can be computed from the joint probability density function in \eqref{joint_pdf_dist}, which results in the following integral:
\begin{equation}
\mathbb{P\left[A_{UE}\right]}= \int_0^\infty \int_0^{\min\left(X_{21}, M^{\frac{1}{\alpha\mu}}\left(\frac{N}{P}\right)^{\frac{2-\mu}{\alpha\mu}}X_{21}^{\frac{2}{\mu}}\right)}{\left(2\pi\lambda_b\right)^2 x_{11}x_{21}\mathrm{e}^{-\pi\lambda_b x_{21}^2}\mathrm{d}x_{11}\mathrm{d}x_{21}}. \nonumber
\end{equation}
By solving the inner integral with respect to $x_{11}$, the expression in \eqref{TIN_prob_received} is obtained. This concludes the proof.
\end{proof}

It is worth noting that $\mathbb{P\left[A_{UE}\right]}=1$ if $M=1$ and $\mu=2$. This corresponds to the scenario where the TIN optimality conditions are inactive and the system model reduces to a conventional cellular network without TIN.
\subsection{Probability Density Function of $X_{11}$ Conditioned Upon $\mathbb{A_{UE}}$}
In this section, we compute the probability density function of the distance, $X_{11}$, between the typical BS and the typical UE conditioned upon the event $\mathbb{A_{UE}}$, i.e., the TIN optimality conditions are fulfilled. It is formally stated in the following lemma.
\begin{lem}
The probability density function of $X_{11}$ conditioned upon $\mathbb{A_{UE}}$ in \eqref{TIN_cell_received} can be formulated as follows:
\begin{equation}
f_{X11}\left(x_{11}|\mathbb{A_{UE}}\right)=\frac{2\pi\lambda_{b}x_{11}\mathrm{e^{-\pi\lambda_{b}\max^2\left(x_{11},x_{11}^{\mu/2}\left(\frac{P}{N}\right)^{\frac{2-\mu}{2\alpha}}\left(\frac{1}{M}\right)^{\frac{1}{2\alpha}}\right)}}}{\mathbb{P\left[A_{UE}\right]}}
\label{pdf_X11},
\end{equation}
where $\mathbb{P\left[A_{UE}\right]}$ is given in Lemma \ref{prob_TIN_lem}.
\label{pdf_X11_lem}
\end{lem}
\begin{proof}
The probability density function of $X_{11}$ conditioned upon $\mathbb{A_{UE}}$ is defined as follows:
\begin{equation}
f_{X11}\left(x_{11}|\mathbb{A_{UE}}\right)= \frac{\mathrm{d}}{\mathrm{d}x_{11}} \frac{\mathbb{P}\left[X_{11}\leq x_{11}, \mathbb{A_{UE}}\right]}{\mathbb{P\left[A_{UE}\right]}}.
\label{temp1}
\end{equation}
The numerator of \eqref{temp1} can be expressed as follows:
\begin{equation}
\begin{split}
\mathbb{P}\left[X_{11}\leq x_{11}, \mathbb{A_{UE}}\right] &= \mathbb{E}_{X_{11}, X_{21}} \left[\mathbf{1}\left(X_{21} > X_{11}\right)\times\mathbf{1}\left(X_{21}> X_{11}^{\mu/2}\left(\frac{P}{N}\right)^{\frac{2-\mu}{2\alpha}}\left(\frac{1}{M}\right)^{1/{2\alpha}}\right)\right]
\\& = \int_0^{x_{11}} \int_{\max\left(x_{11}, x_{11}^{\mu/2}\left(\frac{P}{N}\right)^{\frac{2-\mu}{2\alpha}}\left(\frac{1}{M}\right)^{1/{2\alpha}}\right)}^\infty \left(2\pi \lambda_b\right)^2 x_{11} x_{21}\mathrm{e}^{-\pi\lambda_b x_{21}^2}\mathrm{d}x_{21}\mathrm{d}x_{11},
\end{split}
\end{equation}
where the last equality is obtained by using \eqref{joint_pdf_dist}. The proof follows by computing the inner integral with respect to $x_{21}$ and then applying Leibniz's integration rule. This concludes the proof.
\end{proof}

It is worth noting that \eqref{pdf_X11} reduces to the probability density function of the distance of the nearest BS to the origin of a conventional cellular network if $M=1$ and $\mu=2$. This corresponds to the scenario where TIN is not applied.
\subsection{SINR Coverage Probability} \label{coverage}
In this section, we provide a tractable analytical framework for computing the SINR coverage probability of cellular networks in which the TIN-based scheduling algorithm is applied.

By capitalizing on the three approximations in Section \ref{section_approximation}, the SINR at the typical UE can be formulated as follows:
\begin{equation}
\mathrm{SINR}=\frac{h_{\mathrm{11}}X_{\mathrm{11}}^{-\alpha}}{\sum_{i\in\Phi_{b}^{'}}h_{i}D_{i}^{-\alpha}\mathbf{1}\left(D_{i}\geq\max(X_{11},X_{11}^{\frac{\mu}{2}}\left(\frac{P}{N}\right)^{\frac{2-\mu}{2\alpha}}\left(\frac{1}{M}\right)^{\frac{1}{2\alpha}})\right)+\frac{N}{P}},
\label{sinr_exp}
\end{equation}
where $h_{i}$ is the channel gain of the $i$th interfering BS, $h_{11}$ is the channel gain of the intended link, $D_{i}$ is the distance between the $i$th interfering BS and the typical UE, $\Phi_{b}^{'}$ is the inhomogeneous PPP of interfering BSs whose density is defined in \eqref{den_interf_BS}, and the indicator function makes explicit that the interfering BSs must lie outside the inhomogeneity radius defined in \eqref{interf_boundary}.

We are interested in computing the \textit{effective} SINR coverage probability, $\mathcal{C}_{net}$, of the typical UE, which accounts for the fact that the typical UE may not be served by the typical BS if it is turned off because it does not fulfill the TIN optimality conditions. $\mathcal{C}_{net}$ can be formulated as follows:
\begin{equation}
\mathcal{C}_{net} = \mathbb{P\left[A_{UE}\right]} \mathcal{C},
\label{net_sinr_cov}
\end{equation}
where $\mathbb{P\left[A_{UE}\right]}$ is the  probability that the typical UE is active, which is given in \eqref{TIN_prob_received}, and $\mathcal{C}$ is the SINR coverage probability of the typical active UE. This latter probability is defined and computed in the following theorem.

\begin{thm} \label{TheoCov}
Let $\Theta$ be the SINR decoding threshold. The SINR coverage probability of the typical active UE, ${\mathcal{C}} = {\mathbb{P}}\left[ {{\rm{SINR}} \ge \Theta } \right]$, can be formulated as follows:
\begin{equation}
\mathcal{C}=\frac{2\pi\lambda_{b}}{\mathbb{P}\left[\mathbb{A_{UE}}\right]}\int_{0}^{\infty}x_{11}\mathrm{e^{-\pi\lambda_{b}\max^{2}\left(x_{11},x_{11}^{\mu/2}\left(\frac{P}{N}\right)^{\frac{2-\mu}{2\alpha}}\left(\frac{1}{M}\right)^{\frac{1}{2\alpha}}\right)}\mathrm{e^{\left(-\frac{x_{11}^{\alpha}\Theta N}{P}\right)}}}\mathbb{\mathcal{L_{I}}}\left(x_{11}^{\alpha}\Theta\right)\mathrm{d}x_{11},
\label{sinr_coverage}
\end{equation}
where $\mathbb{\mathcal{L_{I}}}\left(s\right)$ is the Laplace transform of the interference:
\begin{multline}
\mathbb{\mathcal{L_{I}}}\left(s\right)= \mathrm{e}^{\left(\frac{-2\pi\lambda_{b}\mathbb{P}\left[\mathbb{A_{UE}}\right]}{\alpha-2}s\left(\max^{2}\left(x_{11},x_{11}^{\mu/2}\left(\frac{P}{N}\right)^{\frac{2-\mu}{2\alpha}}\left(\frac{1}{M}\right)^{\frac{1}{2\alpha}}\right)\right)^{1-\frac{\alpha}{2}}\right.}  \\
   {}^{\left.{}_{2}\mathrm{F}{}_{1}\left[1,1-\frac{2}{\alpha},2-\frac{2}{\alpha};-s\left(\max^{2}\left(x_{11},x_{11}^{\mu/2}\left(\frac{P}{N}\right)^{\frac{2-\mu}{2\alpha}}\left(\frac{1}{M}\right)^{\frac{1}{2\alpha}}\right)\right)^{-\frac{\alpha}{2}}\right]\right)}.
\label{sinr_coverage_Laplace2}
\end{multline}
\end{thm}
\begin{proof}
See Appendix \ref{proof_thm1}.
\end{proof}
\begin{cor}
If $M=1$ and $\mu=2$, the SINR coverage probability in Theorem \ref{TheoCov} simplifies as follows:
\begin{equation}
\mathcal{C}=2\pi\lambda_{b}\int_{0}^{\infty}x_{11}\mathrm{e}^{-\left(\pi\lambda_{b} x_{11}^2+ \frac{x_{11}^{\alpha} \Theta N}{P} + \pi\lambda_{b}x_{11}^2 \Theta^{\frac{2}{\alpha}} \int\displaylimits_{\Theta^{\frac{-2}{\alpha}}}^{\infty}\frac{1}{1+z^{\alpha/2}}\mathrm{d}z \right)}\mathrm{d}x_{11},
\label{cor1_eq}
\end{equation}
which is the SINR coverage probability of a conventional cellular network when TIN is not applied \cite{Geff_DL}.
\label{cor1}
\end{cor}

\subsection{SINR Average Rate}
In this section, we provide a tractable analytical framework for computing the average rate of cellular networks in which the TIN-based scheduling algorithm is applied.

Similar to the SINR coverage probability, we are interested in computing the \textit{effective} SINR average rate, $\mathcal{R_{SE}}_{net}$, of the typical UE, which accounts for the fact that the typical UE may not be served by the typical BS if it is turned off because it does not fulfill the TIN optimality conditions. $\mathcal{R_{SE}}_{net}$ can be formulated as follows:
\begin{equation}
\mathcal{R_{SE}}_{net} = \mathbb{P\left[A_{UE}\right]} \mathcal{R_{SE}}
\label{se_network},
\end{equation}
where a similar notation as in \eqref{net_sinr_cov} is used and $\mathcal{R_{SE}}$ is the SINR average rate of the typical active UE. This latter average rate is defined and computed in the following theorem and is measured in nats/sec/Hz.

\begin{thm} \label{TheoRate}
Let $\mathcal{R_{SE}} = \mathbb{E}\left[\ln \left[1+\mathrm{SINR}\right]\right]$, where the SINR is given in \eqref{sinr_exp}. The SINR average rate of the typical active UE, $\mathcal{R_{SE}}$, can be formulated as follows:
\begin{equation}
\mathcal{R_{SE}}\!=\!\frac{2\pi\lambda_{b}}{\mathbb{P}\left[\mathbb{A_{UE}}\right]}\!\intop_{0}^{\infty}\!x_{11}\mathrm{e}^{\left(-\pi\lambda_{b}\max^{2}\left(x_{11},x_{11}^{\mu/2}\left(\frac{P}{N}\right)^{\frac{2-\mu}{2\alpha}}\left(\frac{1}{M}\right)^{\frac{1}{2\alpha}}\right)\right)}\!\!\intop_{\tau>0}^{\infty}\!\!\mathrm{e}^{\left(-\frac{N}{P}x_{11}^{\alpha}\left(\mathrm{e^{\tau}-1}\right)\right)}\mathbb{\mathcal{L_{I}}}\left(x_{11}^{\alpha}\left(\mathrm{e^{\tau}-1}\right)\right)\mathrm{d}\tau\mathrm{d}x_{11},
\label{SE_main}
\end{equation}
where $\mathbb{\mathcal{L_{I}}}\left(\cdot\right)$ is the Laplace transform of the interference:
\begin{multline}
\mathbb{\mathcal{L_{I}}}\left(x_{11}^{\alpha}\left(\mathrm{e^{\tau}-1}\right)\right)=\mathrm{e}^{\left(\frac{-2\pi\lambda_{b}\mathbb{P}\left[\mathbb{A_{UE}}\right]}{\alpha-2}x_{11}^{\alpha}\left(\mathrm{e^{\tau}-1}\right)\left(\max^{2}\left(x_{11},x_{11}^{\mu/2}\left(\frac{P}{N}\right)^{\frac{2-\mu}{2\alpha}}\left(\frac{1}{M}\right)^{\frac{1}{2\alpha}}\right)\right)^{1-\frac{\alpha}{2}}\right.}\\ {}^{\left.{}_{2}\mathrm{F}{}_{1}\left[1,1-\frac{2}{\alpha},2-\frac{2}{\alpha};-x_{11}^{\alpha}\left(\mathrm{e^{\tau}-1}\right)\left(\max^{2}\left(x_{11},x_{11}^{\mu/2}\left(\frac{P}{N}\right)^{\frac{2-\mu}{2\alpha}}\left(\frac{1}{M}\right)^{\frac{1}{2\alpha}}\right)\right)^{-\frac{\alpha}{2}}\right]\right)}
\label{spectral_efficiency_thm}
\end{multline}
\end{thm}

\begin{proof}
By definition of SINR average rate, we have:
\begin{equation}
\begin{split}
&\mathcal{R_{SE}}\!\! = \!\!\mathbb{E}\left[\ln \left[1+\mathrm{SINR}\right]\right]\! \!=\!\! \frac{2\pi\lambda_{b}}{\mathbb{P}\left[\mathbb{A_{UE}}\right]}\int_0^\infty \!\!\! \mathbb{E}\left[\!\ln\! \left[1+\frac{h_{\mathrm{11}}x_{\mathrm{11}}^{-\alpha}}{I+\frac{N}{P}}\right]\!\right] \!\!x_{11}\mathrm{e}^{-\pi\lambda_{b}\max^{2}\left(\!x_{11},x_{11}^{\mu/2}\left(\frac{P}{N}\right)^{\frac{2-\mu}{2\alpha}}\left(\frac{1}{M}\right)^{\frac{1}{2\alpha}}\!\right)}\!\mathrm{d}x_{11}\\
&=\frac{2\pi\lambda_{b}}{\mathbb{P}\left[\mathbb{A_{UE}}\right]}\int_0^\infty x_{11}\mathrm{e}^{-\pi\lambda_{b}\max^{2}\left(x_{11},x_{11}^{\mu/2}\left(\frac{P}{N}\right)^{\frac{2-\mu}{2\alpha}}\left(\frac{1}{M}\right)^{\frac{1}{2\alpha}}\right)} \int_{\tau>0}^\infty \mathbb{P}\left[\ln \left[1+\frac{h_{\mathrm{11}}x_{\mathrm{11}}^{-\alpha}}{I+\frac{N}{P}}\right]\geq \tau \right] \mathrm{d}\tau\mathrm{d}x_{11}\\
&=\frac{2\pi\lambda_{b}}{\mathbb{P}\left[\mathbb{A_{UE}}\right]}\int_0^\infty x_{11}\mathrm{e}^{-\pi\lambda_{b}\max^{2}\left(x_{11},x_{11}^{\mu/2}\left(\frac{P}{N}\right)^{\frac{2-\mu}{2\alpha}}\left(\frac{1}{M}\right)^{\frac{1}{2\alpha}}\right)} \int_{\tau>0}^\infty \mathrm{e}^{-x_{11}^\alpha \left(\mathrm{e}^\tau - 1\right)\frac{N}{P}} \mathcal{L}_I\left(x_{11}^\alpha \left(\mathrm{e}^\tau - 1\right)\right) \mathrm{d}\tau\mathrm{d}x_{11},
\end{split}
\end{equation}
where $I$ denotes the other-cell interference and its Laplace transform $\mathcal{L}_I\left(x_{11}^\alpha \left(\mathrm{e}^\tau - 1\right)\right)$ can be computed by using the same analytical steps as for the SINR coverage probability in Theorem \ref{TheoCov}. This concludes the proof.
\end{proof}
\begin{cor}
If $M=1$ and $\mu=2$, the SINR average rate in Theorem \ref{TheoRate} simplifies as follows:
\begin{equation}
\mathcal{R_{SE}}=2\pi\lambda_{b}\intop_{0}^{\infty}x_{11}\mathrm{e}^{\left(-\pi\lambda_{b}x_{11}^2\right)}\intop_{\tau>0}^{\infty}\mathrm{e}^{-\left(\frac{N}{P}x_{11}^{\alpha}\left(\mathrm{e^{\tau}-1}\right)+\pi\lambda_{b}x_{11}^{2}\left(\mathrm{e^{\tau}-1}\right)^{\frac{2}{\alpha}}\int_{\left(\mathrm{e^{\tau}-1}\right)^{-\frac{2}{\alpha}}}^{\infty}\frac{1}{1+z^{\alpha/2}}\mathrm{d}z\right)}\mathrm{d}\tau\mathrm{d}x_{11},
\label{spectral_eff_eq}
\end{equation}
\label{spectral_eff_cor}
which is the SINR average rate of a conventional cellular network if TIN is not applied \cite{Geff_DL}.
\end{cor}
\section{Asymptotic Analysis and System Optimization}
\label{asymptotic_analysis}
The aim of this section is to study the existence and optimal setup for the pair of system parameters $M$ and $\mu$, in order to maximize the effective SINR coverage probability. We focus our attention, in particular, on the effective SINR coverage probability, since the corresponding analytical framework is more tractable than \eqref{SE_main}. 
%
%

The effective SINR coverage probability can be formulated as follows:
\begin{equation}
 \mathcal{C}_{net} = 2\pi\lambda_{b}\int_{0}^{\infty}x_{11}\mathrm{e}^{-\pi\lambda_{b}\max^{2}\left(x_{11},x_{11}^{\mu/2}\left(\frac{P}{N}\right)^{\frac{2-\mu}{2\alpha}}\left(\frac{1}{M}\right)^{\frac{1}{2\alpha}}\right)}\mathrm{e}^{\left(-\frac{sN}{P}\right)}\mathbb{\mathcal{L_{I}}}\left(s\right)\mathrm{d}x_{11},
\label{eff_coverage}
\end{equation}
where  $s=x_{11}^{\alpha}\Theta$, and $\mathbb{\mathcal{L_{I}}}\left(\cdot\right)$ is given in \eqref{sinr_coverage_Laplace2}.

For simplicity, and without loss of generality, we assume $M=1$ and focus our attention on optimizing $\mu \in [1,2]$.
In order to find the optimal value of $\mu$, it is convenient to have a closed form solution of the integral in \eqref{eff_coverage}. To this end, we employ the following approach.

First, we rewrite the $\max(\cdot,\cdot)$ function, which allows us to split the integration range as follows:
\begin{multline}
\mathcal{C}_{net} = 2\pi\lambda_{b}\int_{0}^{\beta^{\frac{1}{\alpha}}}x_{11}\mathrm{e}^{\left(-\pi\lambda_{b}x_{11}^{\mu}\beta^{\frac{2-\mu}{\alpha}}\right)}\mathrm{e}^{\left(-\frac{x_{11}^{\alpha}\Theta}{\beta}\right)}\mathrm{e}^{\left(-\pi\lambda_{b}\mathbb{P}\left[\mathbb{A_{UE}}\right]x_{11}^2 \Theta^{\frac{2}{\alpha}}\int\displaylimits_{x_{11}^{\mu-2}\Theta^{\frac{-2}{\alpha}}\beta^{\frac{2-\mu}{\alpha}}}^{\infty}\frac{1}{1+z^{\alpha/2}}\mathrm{d}z\right)}\mathrm{d}x_{11}
\\ +
2\pi\lambda_{b}\int_{\beta^{\frac{1}{\alpha}}}^{\infty}x_{11}\mathrm{e}^{\left(-\pi\lambda_{b}x_{11}^2\right)}\mathrm{e}^{\left(-\frac{x_{11}^{\alpha}\Theta}{\beta}\right)}\mathrm{e}^{\left(-\pi\lambda_{b}\mathbb{P}\left[\mathbb{A_{UE}}\right]x_{11}^2 \Theta^{\frac{2}{\alpha}}\int\displaylimits_{\Theta^{\frac{-2}{\alpha}}}^{\infty}\frac{1}{1+z^{\alpha/2}}\mathrm{d}z\right)}\mathrm{d}x_{11}
\label{eff_cov_splitt},
\end{multline}
where we introduce the notation $\beta=\frac{P}{N}$. 

In \eqref{eff_cov_splitt}, the second integral is negligible compared with the first integral for sufficiently high values of the SNR $\beta$ (high SNR regime). Under this assumption, \eqref{eff_cov_splitt} can be approximated as follows:
\begin{equation}
\mathcal{C}_{net} \approx 2\pi\lambda_{b}\int_{0}^{\beta^{\frac{1}{\alpha}}}x_{11}\mathrm{e}^{\left(-\pi\lambda_{b}x_{11}^{\mu}\beta^{\frac{2-\mu}{\alpha}}\right)}\mathrm{e}^{\left(-\frac{x_{11}^{\alpha}\Theta}{\beta}\right)}\mathrm{e}^{\left(-\pi\lambda_{b}\mathbb{P}\left[\mathbb{A_{UE}}\right]x_{11}^2 \Theta^{\frac{2}{\alpha}}\int\displaylimits_{x_{11}^{\mu-2}\Theta^{\frac{-2}{\alpha}}\beta^{\frac{2-\mu}{\alpha}}}^{\infty}\frac{1}{1+z^{\alpha/2}}\mathrm{d}z\right)}\mathrm{d}x_{11}
\label{eff_single}.
\end{equation}

It is known that the integral inside the exponential function can be expressed in terms of the Gauss hypergeometric function for general values of the path-loss exponent $\alpha >2$ \cite{Geff_DL}. In order to better highlight the proposed approach, we assume $\alpha=4$ in the sequel. A similar approach can be applied for other values of $\alpha$. This generalization is left to the reader. By letting $\alpha=4$ and by using some algebraic manipulations, \eqref{eff_single} can be written as follows:
\begin{equation}
\mathcal{C}_{net} \approx 2\pi\lambda_{b}\int_{0}^{\beta^{\frac{1}{4}}}x_{11}\mathrm{e}^{\left(-\pi\lambda_{b}x_{11}^{\mu}\beta^{\frac{2-\mu}{4}}\right)}\mathrm{e}^{\left(-\frac{x_{11}^{4}\Theta}{\beta}\right)}\mathrm{e}^{\left(-\pi\lambda_{b}\mathbb{P}\left[\mathbb{A_{UE}}\right]x_{11}^2 \sqrt{\Theta}\left(\frac{\pi}{2}-\arctan\left(\frac{x_{11}^{\mu-2}\beta^{\frac{2-\mu}{4}}}{\sqrt{\Theta}}\right)\right)\right)}\mathrm{d}x_{11}
\label{eff_single_a4}.
\end{equation}

The term $\mathrm{e}^{\left(-\frac{x_{11}^{4}\Theta}{\beta}\right)}$ can be ignored, because when $\beta$ is dominant then $\mathrm{e}^{\left(-\frac{x_{11}^{4}\Theta}{\beta}\right)}\approx 1$ and when the $x_{11}^{4}\Theta$ dominates then the rest of the integral tends to $0$. Therefore, \eqref{eff_single_a4} simplifies as follows:
\begin{equation}
\mathcal{C}_{net} \approx 2\pi\lambda_{b}\int_{0}^{\beta^{\frac{1}{4}}}x_{11}\mathrm{e}^{\left(-\pi\lambda_{b}x_{11}^{\mu}\beta^{\frac{2-\mu}{4}}-\pi\lambda_{b}\mathbb{P}\left[\mathbb{A_{UE}}\right]x_{11}^2 \sqrt{\Theta}\left(\frac{\pi}{2}-\arctan\left(\frac{x_{11}^{\mu-2}\beta^{\frac{2-\mu}{4}}}{\sqrt{\Theta}}\right)\right)\right)}\mathrm{d}x_{11}
\label{eff_single_a4_N0}.
\end{equation}

In the next two sections, we further simplify \eqref{eff_single_a4_N0} by considering large and small values of $\Theta$, respectively.
\subsection{Large Values of $\Theta$}
By definition, $\mu \in [1,2]$. In addition, the function $\arctan\left(\cdot\right)$ becomes small for large values of $\Theta$, i.e., $\arctan\left(1/\Theta\right) \approx 0$. Thus, it can be ignored for large values of $\Theta$:
\begin{equation}
\mathcal{C}_{net} \approx 2\pi\lambda_{b}\int_{0}^{\beta^{\frac{1}{4}}}x_{11}\mathrm{e}^{\left(-\pi\lambda_{b}x_{11}^{\mu}\beta^{\frac{2-\mu}{4}}-\frac{\pi^2\lambda_{b}\mathbb{P}\left[\mathbb{A_{UE}}\right]x_{11}^2 \sqrt{\Theta}}{2}\right)}\mathrm{d}x_{11} = 2\pi\lambda_{b}\int_{0}^{\beta^{\frac{1}{4}}}x_{11}\mathrm{e}^{\left(-A_1 x_{11}^{\mu}-A_2x_{11}^2 \right)}\mathrm{d}x_{11}
\label{eff_single_a4_N0_atan0},
\end{equation}
where $A_1 = \pi \lambda_b \beta^{\frac{2-\mu}{4}}$ and $A_2= \frac{\pi^2\lambda_b \mathbb{P}\left[\mathbb{A_{UE}}\right] \sqrt{\Theta} }{2}$. 

The integral in \eqref{eff_single_a4_N0_atan0} cannot be formulated in a simple closed-form expression that is suitable to get insight for system optimization. To circumvent this issue, we express the exponential $\mathrm{e}^{\left(-A_1 x_{11}^u\right)}$ by using its power series representation:
\begin{equation}
\mathcal{C}_{net} \approx 2\pi\lambda_b \int_{0}^{\beta^{\frac{1}{4}}} x_{11}\left[1-\frac{A_1 x_{11}^{\mu}}{1!}+\frac{A_1^2 x_{11}^{2\mu}}{2!}-\frac{A_1^3 x_{11}^{3\mu}}{3!}+\frac{A_1^4 x_{11}^{4\mu}}{4!}+\cdots\right]\mathrm{e}^{\left(-A_2 x_{11}^2\right)}\mathrm{d}x_{11}.
\end{equation}

To further simplify the analysis, since the high SNR assumption is considered, the upper integration limit can be simplified as $\beta^{1/\alpha} \to \infty$. With the aid of this approximation, we obtain:
\begin{multline}
\mathcal{C}_{net} \approx 2\pi\lambda_b\left[ \int_{0}^{\infty} x_{11}\mathrm{e}^{\left(-A_2 x_{11}^2\right)}\mathrm{d}x_{11}-\frac{A_1 }{1!}\int_{0}^{\infty}x_{11}^{\mu+1}\mathrm{e}^{\left(-A_2 x_{11}^2\right)}\mathrm{d}x_{11}+\frac{A_1^2 }{2!}\int_{0}^{\infty}x_{11}^{2\mu+1}\mathrm{e}^{\left(-A_2 x_{11}^2\right)}\mathrm{d}x_{11} \right. \\ \left.
-\frac{A_1^3 }{3!}\int_{0}^{\infty}x_{11}^{3\mu+1}\mathrm{e}^{\left(-A_2 x_{11}^2\right)}\mathrm{d}x_{11}+\frac{A_1^4 }{4!}\int_{0}^{\infty}x_{11}^{4\mu+1}\mathrm{e}^{\left(-A_2 x_{11}^2\right)}\mathrm{d}x_{11}\cdots\right]\\
\approx \pi\lambda_b\left[\frac{\Gamma\left(1\right)}{A_2}-\frac{A_1 \, \Gamma\left(\frac{\mu+2}{2}\right)}{A_2^{\frac{\mu+2}{2}}}+\frac{A_1^2\,\Gamma\left(\frac{2\mu+2}{2}\right)}{2! A_2^{\frac{2\mu+2}{2}}}-\frac{A_1^3\,\Gamma\left(\frac{3\mu+2}{2}\right)}{3! A_2^{\frac{3\mu+2}{2}}}+\frac{A_1^4\,\Gamma\left(\frac{4\mu+2}{2}\right)}{4! A_2^{\frac{4\mu+2}{2}}}+ \cdots \right]
\label{power_series_ext_gamma1},
\end{multline}
where the last expression is obtained by solving each integral and writing it in terms of the Gamma function, i.e., $\intop_{0}^{\infty}x^m\mathrm{e}^{-\zeta x^n}\mathrm{d}x=\frac{\Gamma\left(\gamma\right)}{n\zeta^{\gamma}}$ where $\gamma=\frac{m+1}{n}$ \cite{book_integral}. 

%

Equation \eqref{power_series_ext_gamma1} can be formulated, in a more compact form, as follows:
\begin{equation}
\mathcal{C}_{net} \approx \frac{\pi\lambda_b}{A_2}\sum_{n=0}^{\infty} \frac{\left(-1\right)^{n} R^{\frac{n}{2}} \Gamma\left(\frac{n\mu+2}{2}\right)}{n!}
\label{compact_1_series},
\end{equation}
where $R=A_1^2/A_2^{\mu}$.

It is not possible, to the best of our knowledge, to compute the explicit result of the series for arbitrary values of $\mu$. This is possible, however, for the special case $\mu=2$:
\begin{equation}
\mathcal{C}_{net} \approx  = \frac{\pi \lambda_b}{A_1 + A_2}\nonumber
\label{convergance_mu2}.
\end{equation}

In Fig. \ref{R_function}, we plot $R$ as a function of $\mu$. We observe that, for constant values of $\lambda_b$ and $\beta$, $R$ decreases if $\mu$ increases. It is difficult, however, to find the analytical expression of the optimal value of $\mu$ that maximizes $\mathcal{C}_{net}$ in \eqref{compact_1_series}. To circumvent this issue, we have performed extensive Monte Carlo simulations and found that ${C}_{net}$ is maximized by the value of $\mu$ that corresponds to $R=1$.
\begin{figure}
	\centering
		\includegraphics[scale=0.7]{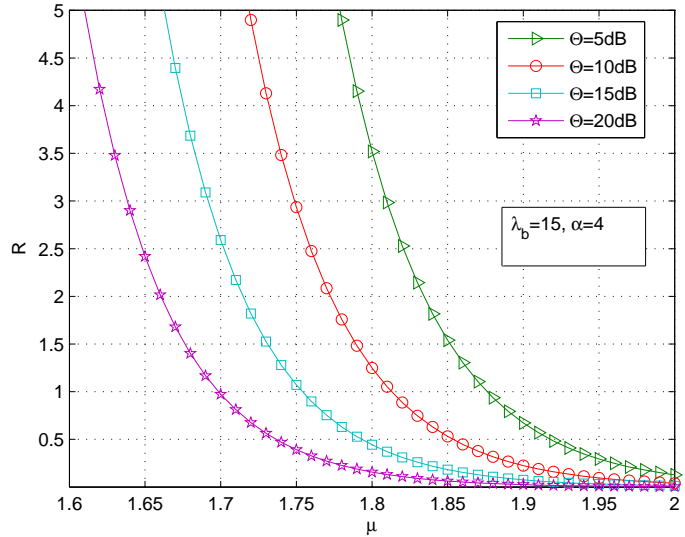}
	\caption{$R$ as a function of $\mu$.}
	\label{R_function}
	\vspace{-1.7em}
\end{figure}

As a result, the optimal value of $\mu$ that maximizes ${C}_{net}$ is the unique solution of the following equation $R=1$:
\begin{equation}
A_1^2-A_2^{\mu}=\left(\pi^2 \lambda_b^2 \beta^{\frac{2-\mu}{2}}\right)-\left(\frac{\pi^2\lambda_b \mathbb{P}\left[\mathbb{A_{UE}}\right] \sqrt{\Theta} }{2}\right)^{\mu}=0.
\label{optimization_simple}
\end{equation}

To compute the optimal value of $\mu$, a tractable expression of $\mathbb{P}\left[\mathbb{A_{UE}}\right]$ is needed, which is itself a function of $\mu$. From \eqref{TIN_prob_received}, $\mathbb{P}\left[\mathbb{A_{UE}}\right]$ can be computed as follows:
\begin{equation}
\mathbb{P}\left[\mathbb{A_{UE}}\right] =2\left(\pi\lambda_{b}\right)^2\int_{0}^{\infty}x_{21}\mathrm{e}^{-\pi\lambda_{b}x_{21}^2}{\min^2} \left(x_{21},\beta^{\frac{\mu-2}{4\mu}}x_{21}^{\frac{2}{\mu}}\right)\mathrm{d}x_{21}\approx \frac{2\, \Gamma\left(\frac{2}{u}\right)}{\mu \left(\pi\lambda_b\right)^{\frac{2}{\mu}-1}\beta^{\frac{1}{\mu}-\frac{1}{2}}}
\label{prob_TIN_approx},
\end{equation}
where the approximation is obtained by using similar approximations as those used for computing ${C}_{net}$ in the high SNR regime.

By inserting \eqref{prob_TIN_approx} in \eqref{optimization_simple} and with the aid of some algebraic manipulations, we obtain:
\begin{equation}
\mu^{\mu}\left(\pi \lambda_b\right)^4\beta^{2-\mu}-\left(\pi^3\lambda_b^2 \sqrt{\Theta}\,  \Gamma\left(\frac{2}{\mu}\right)\right)^{\mu}=0.
\label{optimization_simple_final}
\end{equation}

The optimization problem in \eqref{optimization_simple_final} is much simpler to solve than \eqref{eff_coverage}. For example, it can be easily solved by using the \texttt{fzero} function in Matlab. By direct inspection of \eqref{optimization_simple_final}, in addition, the following conclusions can be drawn. The minuend term, $\mu^{\mu}\left(\pi \lambda_b\right)^4\beta^{2-\mu}$, is independent of $\Theta$, and the subtrahend term, $\left(\pi^3\lambda_b^2 \sqrt{\Theta}\,  \Gamma\left(\frac{2}{\mu}\right)\right)^{\mu}$, is dependent on $\Theta$. In the minuend, the term $\mu^{\mu}$ increases as $\mu$ increases from $1$ to $2$, but the term $\beta^{2-\mu}$ decreases very rapidly with the same increase of $\mu$. This suggests that the minuend is a decreasing function of $\mu$. For fixed and realistic values of $\lambda_b$ and $\Theta$, the subtrahend decreases if $\mu$ increases. We plot both the minuend and subtrahend terms in Fig. \ref{splitting_optimization}. It can be observed that, for the given parameters, both the minuend and subtrahend decrease if $\mu$ increases. The minuend and subtrahend terms, however, cross each other in exactly one point, which guarantees that there is a unique optimal value of $\mu$ that maximizes the effective coverage probability.

\begin{figure}
	\centering
		\includegraphics[scale=0.7]{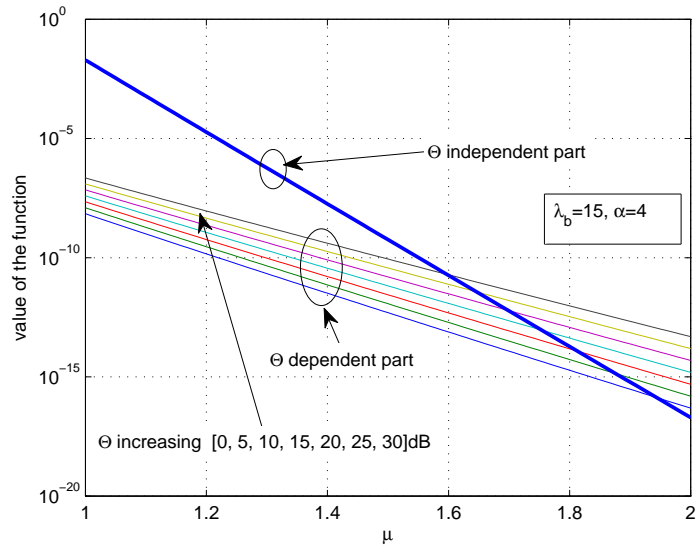}
	\caption{Understanding the optimization problem \eqref{optimization_simple_final}.}
	\label{splitting_optimization}
	\vspace{-1.7em}
\end{figure}
\begin{figure*}
\centering
			\begin{subfigure}{0.45\textwidth}
			\centering
			\includegraphics[width=1\linewidth]{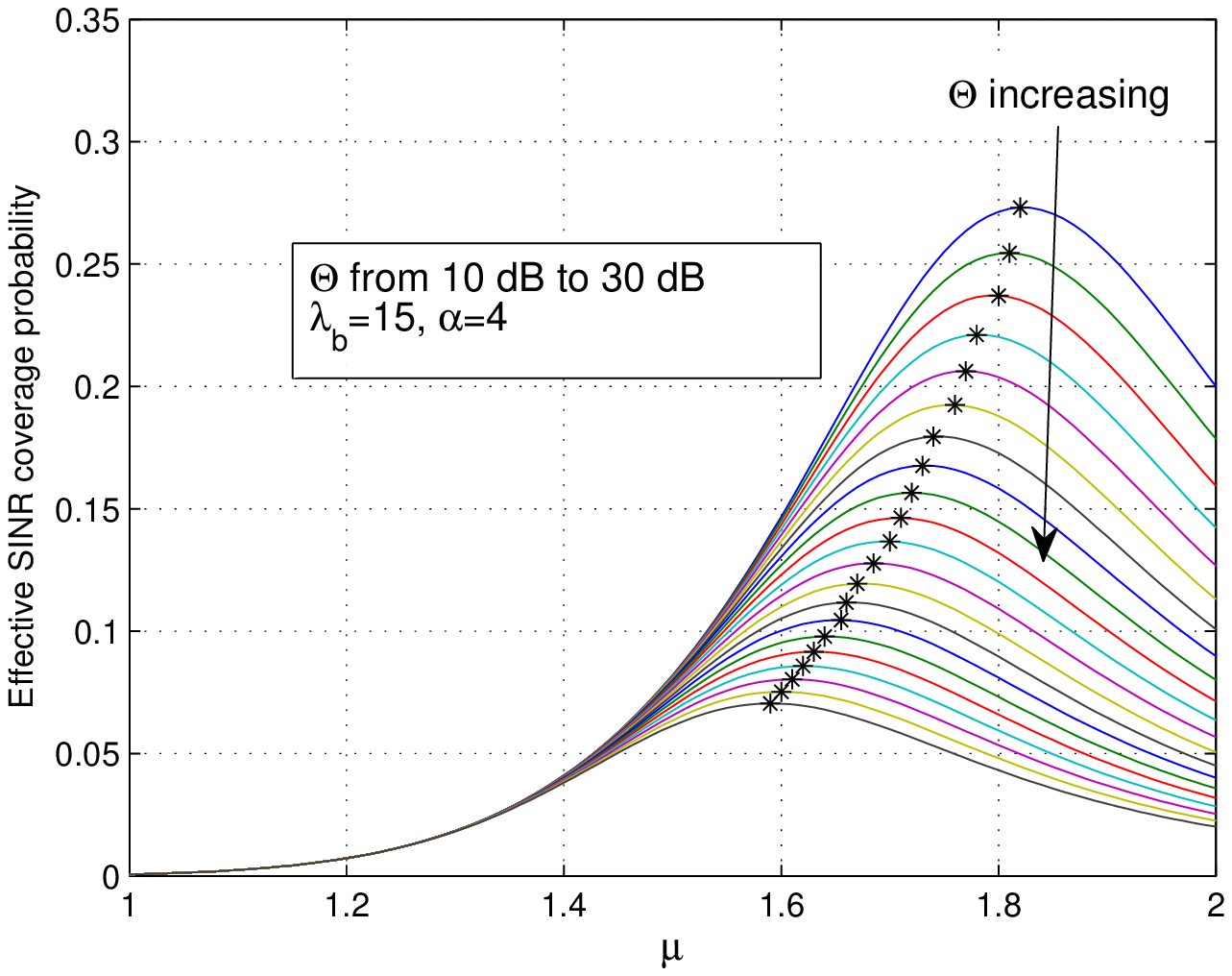}
			\caption{Effective SINR coverage probability.}
			\label{net_cov}
		\end{subfigure}
	\begin{subfigure}{0.45\textwidth}
		\centering
		\includegraphics[width=1\linewidth]{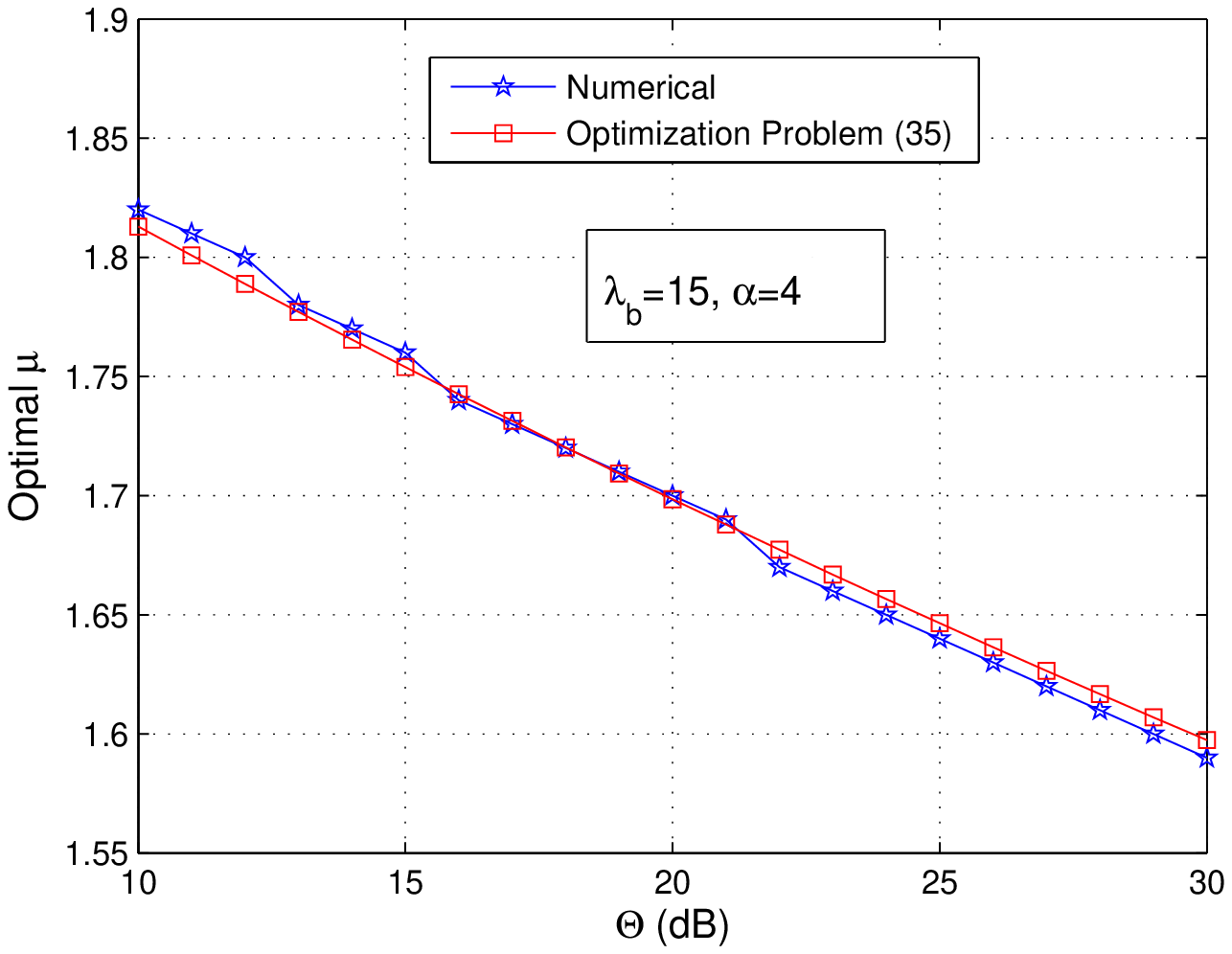}
		 \caption{Optimal $\mu$.}
		 \label{optimal_mu}
	\end{subfigure}\hfill
\caption{Optimal $\mu$ for large SINR threshold $\Theta$.}
		\label{optimal_mu_large_theta}
\end{figure*}

In Fig. \ref{optimal_mu_large_theta}, we show that, for each value of $\Theta$, an optimal value of $\mu$ exists. In addition, the accuracy of the solution of \eqref{optimization_simple_final} compared with the exact values of $\mu$ that maximizes \eqref{eff_coverage} is studied. A good accuracy is obtained. Fig. \ref{net_cov}, in particular, shows that the effective coverage probability is maximized for values of $\mu$ smaller than $2$, which correspond to a conventional cellular network where TIN is not applied. 
\begin{figure}
	\centering
		\includegraphics[scale=0.5]{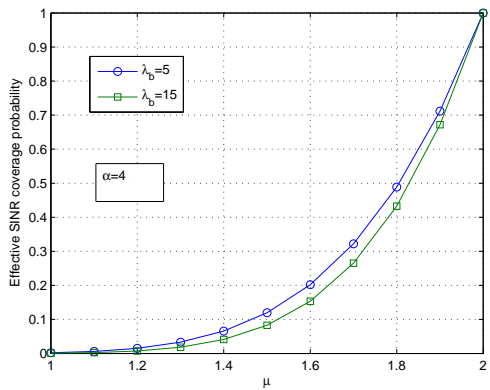}
	\caption{Optimal $\mu$ when $\Theta$ is very small.}
	\label{optimal_mu_small_theta}
		\vspace{-2.7em}
\end{figure}
\subsection{Small Values of $\Theta$}
In this section, we study the existence of optimal values of $\mu$ for small values of the SINR threshold. If $\Theta$ is small, i.e. $\Theta\approx0$, we have $\arctan\left(\frac{x_{11}^{\mu-2}\beta^{\frac{2-\mu}{4}}}{\sqrt{\Theta}}\right)=\arctan\left(\infty\right)=\frac{\pi}{2}$, and \eqref{eff_single_a4_N0} reduces to:
\begin{equation}
\mathcal{C}_{net} \approx 2\pi\lambda_{b}\int_{0}^{\infty}x_{11}\mathrm{e}^{\left(-\pi\lambda_{b}x_{11}^{\mu}\beta^{\frac{2-\mu}{4}}\right)}\mathrm{d}x_{11} \approx \frac{2\, \Gamma\left(\frac{2}{u}\right)}{\mu \left(\pi\lambda_b\right)^{\frac{2}{\mu}-1}\beta^{\frac{1}{\mu}-\frac{1}{2}}}
\label{eff_single_a4_N0_small_theta}.
\end{equation}

By direct inspection of \eqref{eff_single_a4_N0_small_theta}, we evince that $\mathcal{C}_{net}$ increases if $\mu$ increases. In particular, $\mathcal{C}_{net}=1$ when $\mu=2$. This finding suggests that there is no need to turn any BSs off if the SINR threshold is small. For small values of $\Theta$, in other words, there is no need to apply TIN. Figure \ref{optimal_mu_small_theta} confirms this conclusion. 
\section{Results and Discussion}
\label{results}
In this section, we illustrate some numerical and simulation results. We emphasize that the Monte Carlo simulation results are generated without any approximations or assumptions that are exploited to obtain the analytical frameworks. As for the simulation setup, we set $P=46$ dBm and $N=-110$ dBm/Hz. 

In Fig. \ref{prob_TIN}, we plot the probability of TIN against $\mu$. The figure shows that the probability of TIN increases with $\mu$. It can be observed that the probability of TIN is zero if $\mu=1$ and tends to one if $M=10$ and $\mu=1.9$. Similarly, the probability tends to one if $M=1$ and $\mu=2$. We remind the reader that no BSs are turned off if the probability of TIN is equal to one. In other words, every BS schedules a UE in any given time slot, as in conventional cellular networks. This illustrates that both $M$ and $\mu$ are tunable parameters that control the number of UEs to be scheduled.
For simplicity, and without loss of generality, we keep $M=1$ and vary $\mu$ in the rest of the results.

\begin{figure*}
\centering
			\begin{subfigure}{0.45\textwidth}
			\centering
			\includegraphics[width=1\linewidth]{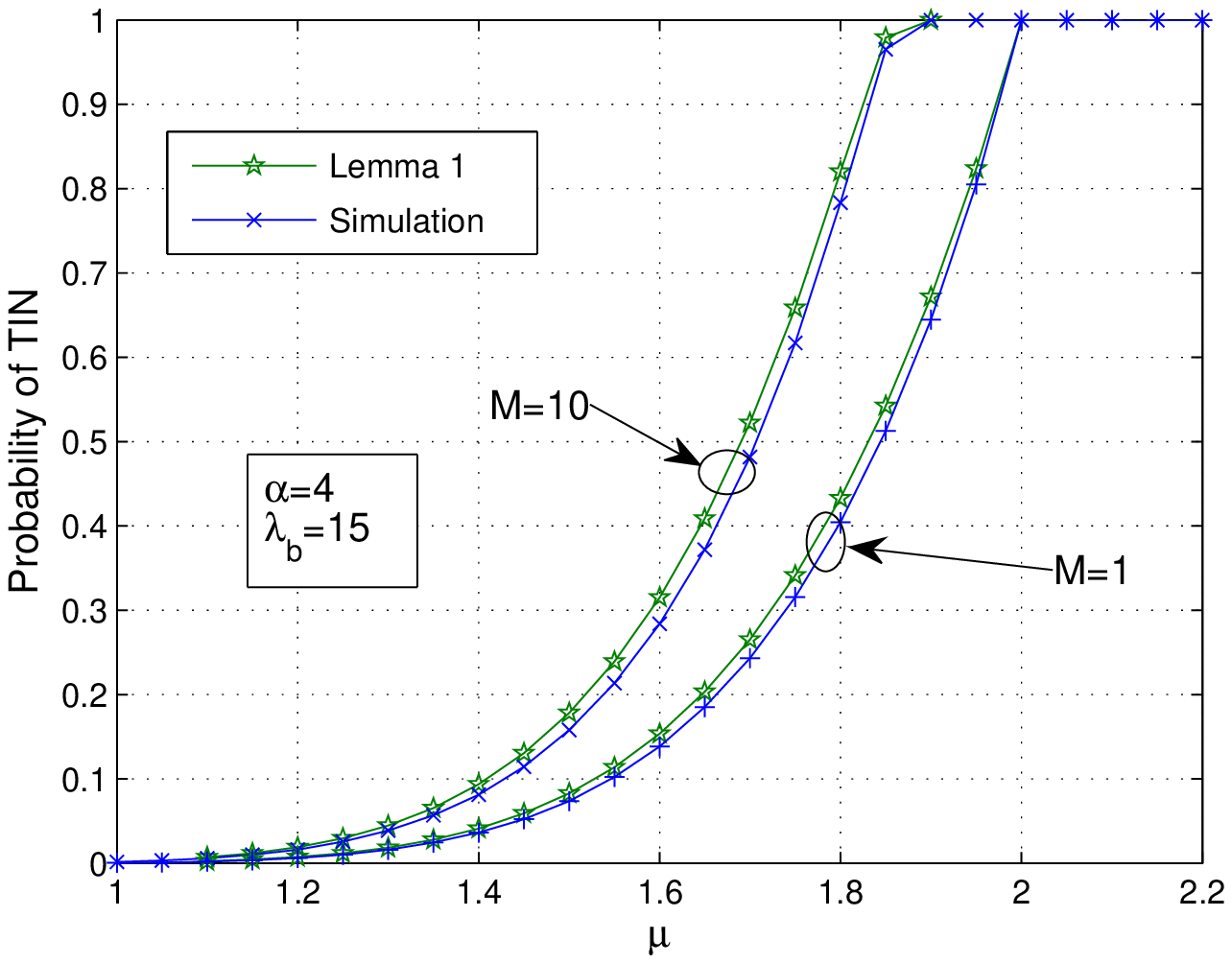}
			\caption{}
			\label{}
		\end{subfigure}
	\begin{subfigure}{0.45\textwidth}
		\centering
		\includegraphics[width=1\linewidth]{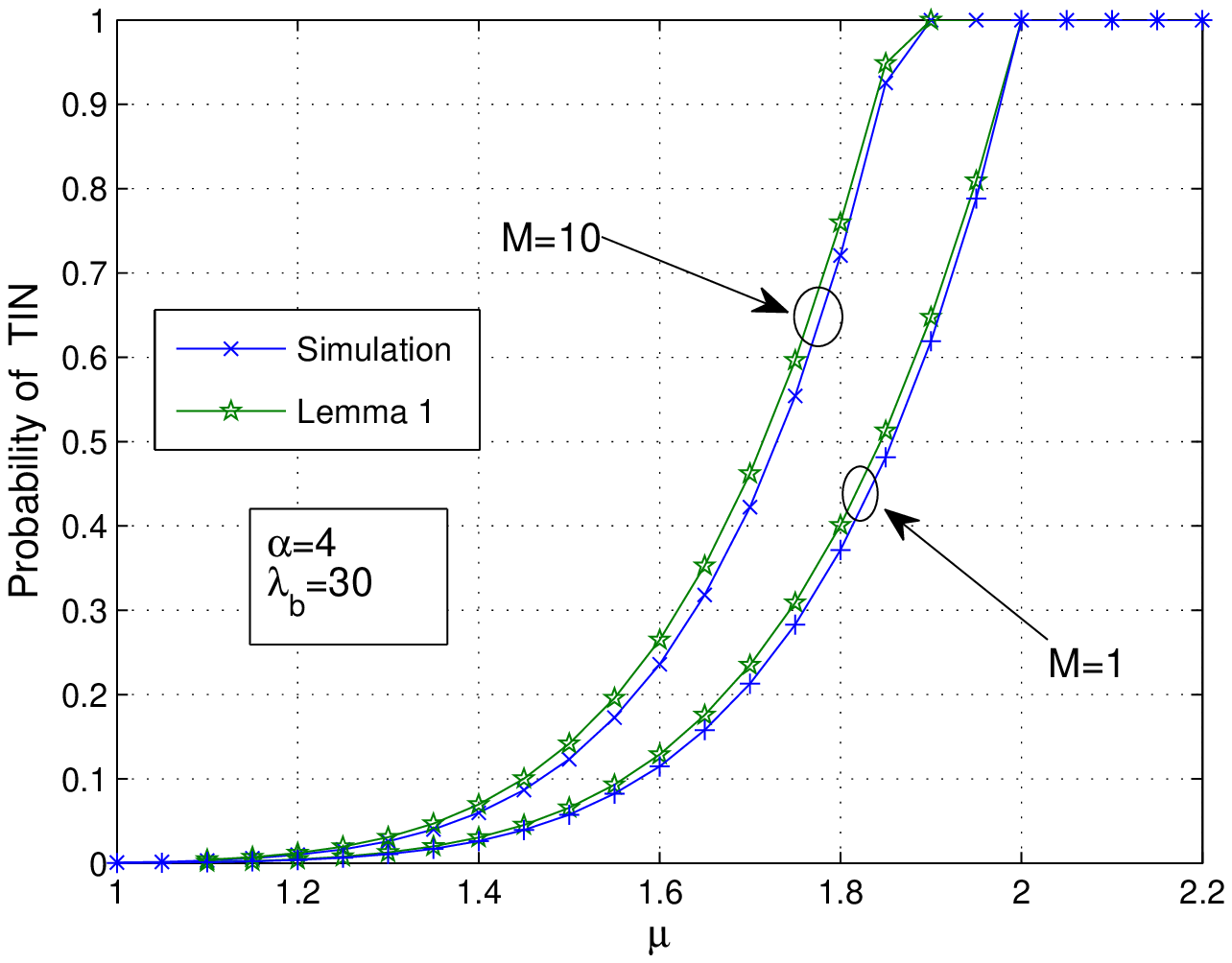}
		 \caption{}
		 \label{}
	\end{subfigure}\hfill
\caption{Probability of TIN versus $\mu$.}
		\label{prob_TIN}
\end{figure*}
\begin{figure}
	\centering
		\includegraphics[scale=0.65]{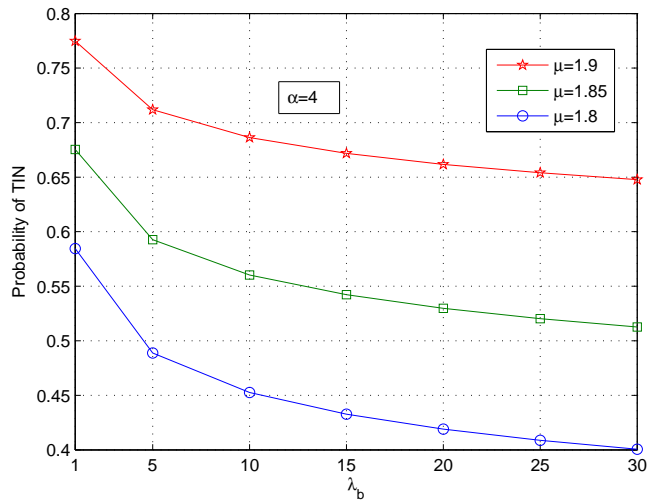}
	\caption{Probability of TIN versus $\lambda_b$.}
	\label{prob_TIN_vs_lambda_b}
	\vspace{-1.7em}
\end{figure}

In Fig. \ref{prob_TIN_vs_lambda_b}, we plot the probability of TIN versus the density of the BS. It can be observed that, for a given value of $\mu$, the probability of TIN decreases if $\lambda_b$ increases. This result shows that, if we increase $\lambda_b$, the probability that the UEs satisfy the TIN criterion in \eqref{TIN_cell_received} reduces. This is due to the increase of the amount of interference in the network.
\begin{figure}
	\centering
		\includegraphics[scale=0.65]{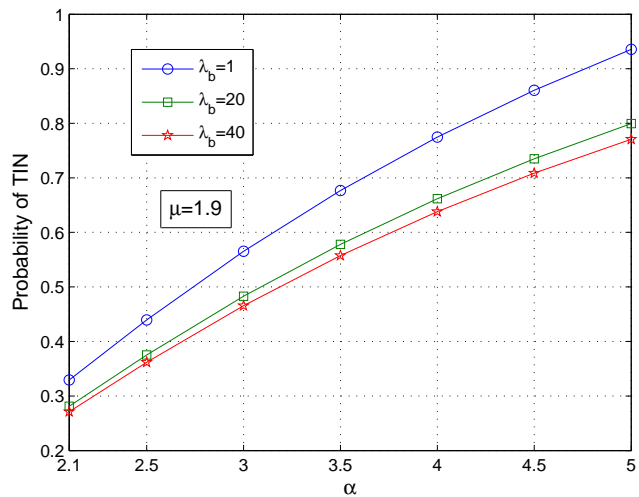}
	\caption{Probability of TIN versus $\alpha$.}
	\label{prob_TIN_vs_pathloss}
	\vspace{-1.7em}
\end{figure}

Fig. \ref{prob_TIN_vs_pathloss} shows the effect of the path-loss exponent on the probability of TIN. It can be observed that the probability of TIN increases if $\alpha$ increases. If the path-loss exponent is large, the interference received at a UE is low, and, therefore, the probability that a UE satisfies the TIN criterion increases \eqref{TIN_cell_received}.

Fig. \ref{sinr_cov} compares the simulation and analytical curves of the SINR coverage probability. We remind the reader that the first simulation curve is obtained if the TIN-based scheduling is based on \eqref{TIN_cell} and the second simulation curve is obtained if the TIN-based scheduling is based on \eqref{TIN_cell_received}. The curve corresponding to the classical scheduling is obtained by not turning any BSs off.
The figure highlights that the analytical and simulation curves are quite close to each other, which substantiates the accuracy of our analysis. The small gap between the two curves originates from the approximations discussed in Section \ref{section_approximation}.
\begin{figure*}
\centering
			\begin{subfigure}{0.45\textwidth}
			\centering
			\includegraphics[width=1\linewidth]{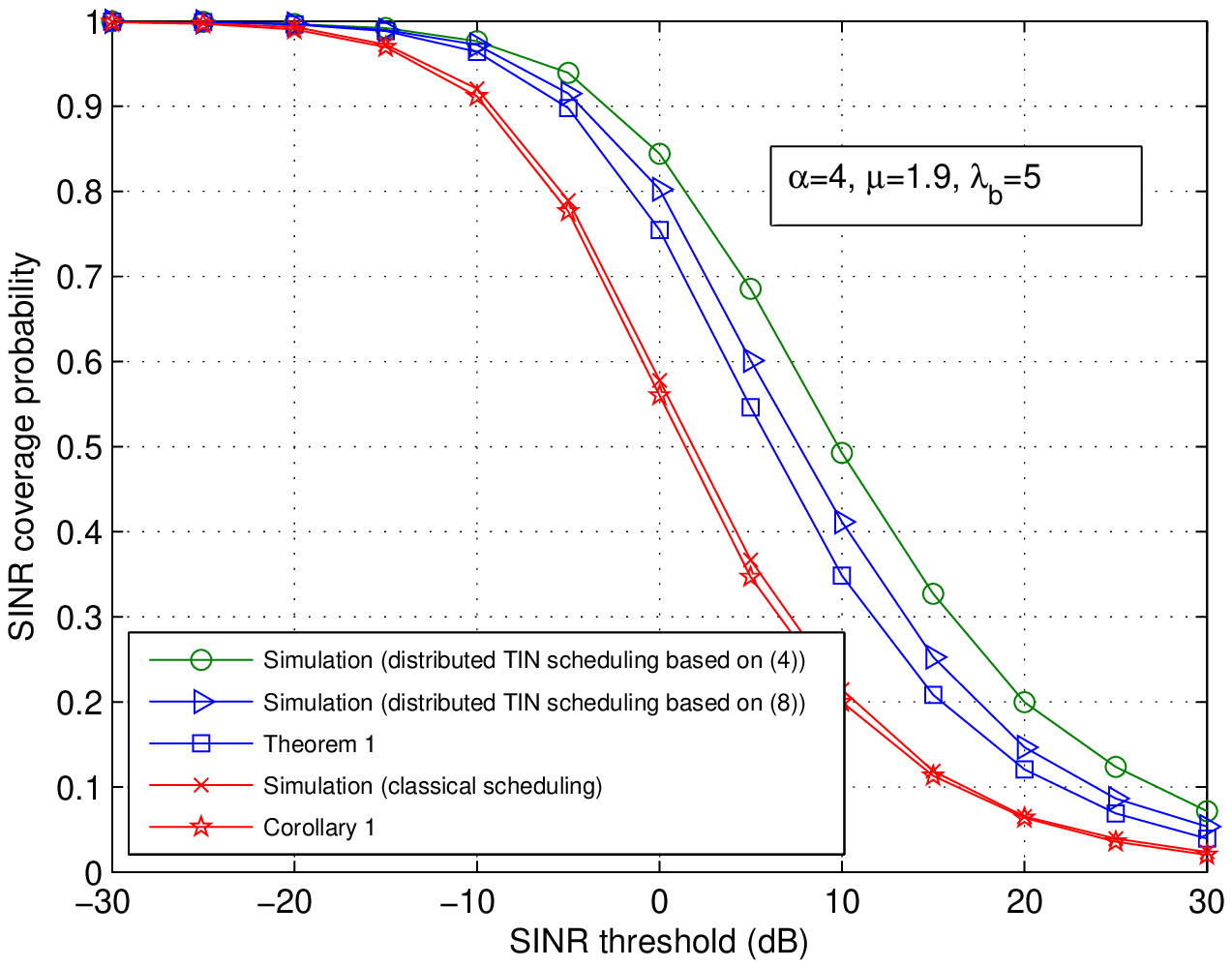}
			\caption{}
			\label{}
		\end{subfigure}
	\begin{subfigure}{0.45\textwidth}
		\centering
		\includegraphics[width=1\linewidth]{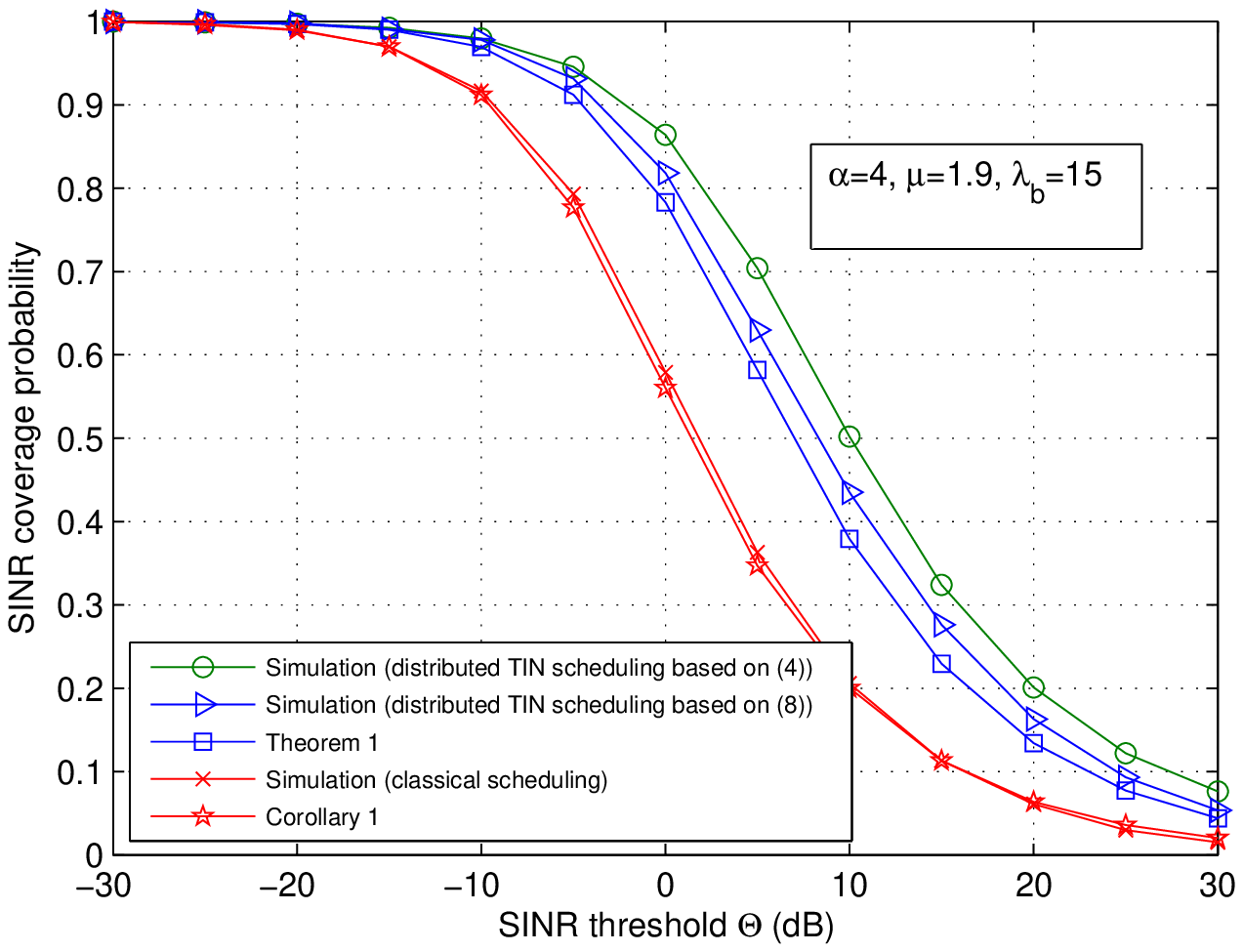}
		 \caption{}
		 \label{}
	\end{subfigure}\hfill
\caption{SINR coverage probability.}
		\label{sinr_cov}
\end{figure*}
\begin{figure}
	\centering
		\includegraphics[scale=0.7]{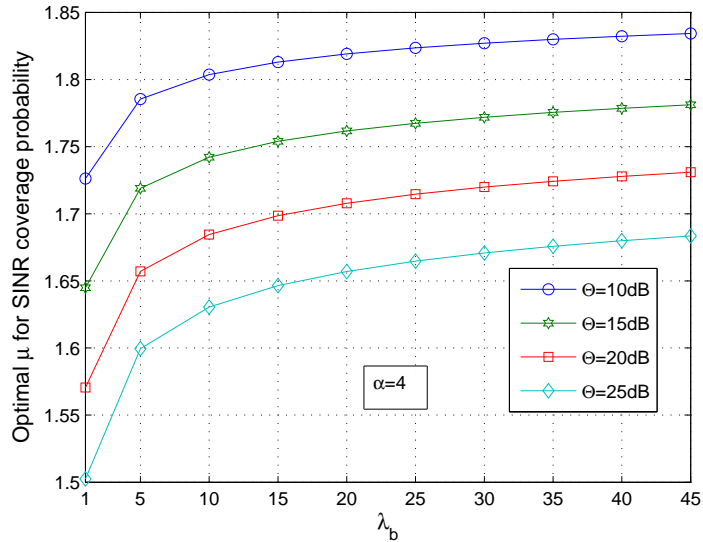}
	\caption{Optimal $\mu$ for the SINR coverage probability.}
	\label{optimal_mu_for_SINR_coverage}
\end{figure}

Fig. \ref{optimal_mu_for_SINR_coverage} illustrates the effect of $\lambda_b$ and $\Theta$ on the optimal value of $\mu$ that maximizes the SINR coverage probability. The optimal value of $\mu$ is obtained as the solution of \eqref{optimization_simple_final}. The figure highlights that the optimal value of $\mu$ decreases if the SINR threshold increases. This shows that, for large values of the SINR threshold, more BSs need to be turned off to maximize the effective coverage probability. Furthermore, it can be observed that the optimal value of $\mu$ increases if the density of the BSs increases. To understand this effect, we need to consider the effect of $\lambda_b$ on the probability of TIN. We have seen in Fig. \ref{prob_TIN_vs_lambda_b} that the probability of TIN decreases if the density of the BS increases. Our optimization problem not only maximizes the SINR coverage probability, but, at the same time, tries to turn the smallest number of BSs off. Therefore, the optimal value of $\mu$ increases by increasing the density of the BS.
\begin{figure}
	\centering
		\includegraphics[scale=0.7]{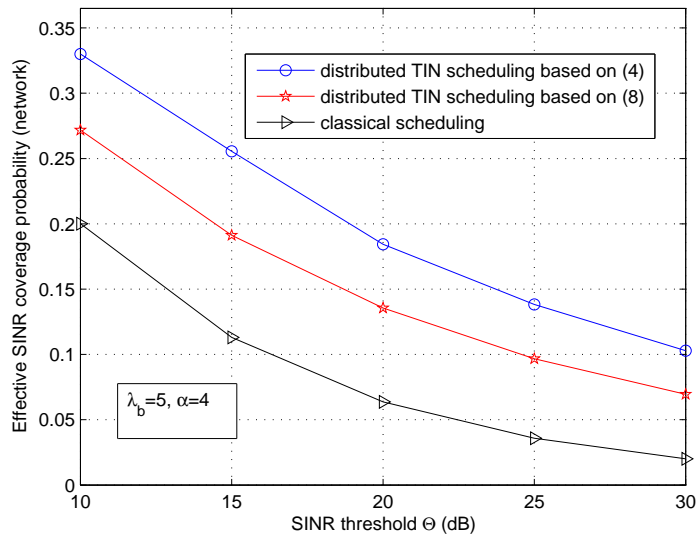}
	\caption{SINR coverage probability gain.}
	\label{gain_sinr}
\end{figure}

Fig. \ref{gain_sinr} shows the SINR coverage probability gain provided by our proposed distributed TIN-based scheduling scheme compared with the conventional scheduling scheme where no BSs are turned off. The optimal value of $\mu$ is obtained from \eqref{optimization_simple_final} if the distributed TIN scheduling algorithm is based on \eqref{TIN_cell_received}. It is obtained through simulations, on the other hand, if the distributed TIN scheduling algorithm is based on \eqref{TIN_cell}. It can be observed that the improvement provided by the distributed TIN scheduling scheme over the conventional scheduling scheme changes with the SINR threshold.
Specifically, it can be noticed that, if $\Theta=10$dB, the improvement is $67\%$ when the distributed TIN scheduling is based on \eqref{TIN_cell}, and $36\%$ when the distributed TIN scheduling is based on \eqref{TIN_cell_received}.
\begin{figure}
	\centering
		\includegraphics[scale=0.7]{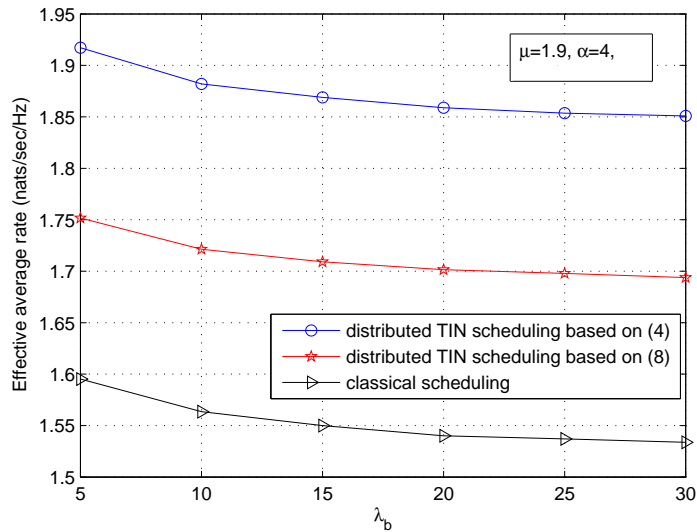}
	\caption{Average rate gain.}
	\label{gain_se}
\end{figure}

In Fig. \ref{gain_se}, we depict the effective average rate against the density of the BSs. The curves of the distributed TIN scheduling are obtained by setting $\mu=1.9$. It can be observed that the distributed TIN scheduling based on \eqref{TIN_cell} improves the average rate by 21$\%$, whereas the distributed TIN scheduling based on \eqref{TIN_cell_received} improves the average rate by 11$\%$. Furthermore, it can be observed that the gain remains constant for various values of $\lambda_b$. From Fig. \ref{gain_sinr} and \ref{gain_se}, it can be concluded that a simple distributed TIN-based scheduling algorithm significantly enhances the SINR effective coverage probability and the effective average rate.
%
%
\section{Conclusion}
\label{conclusion}
In this paper, we have proposed a simple scheduling algorithm for application to cellular networks that is based on the TIN optimality condition. The original form of the scheduling algorithm is shown not to be mathematically tractable. To overcome this issue, we have proposed a simplified analytical framework to estimate the SINR effective coverage probability and the effective average rate by using stochastic geometry tools. To enable a simple optimization of the system parameters, we have developed simplified analytical frameworks in the high SNR regime, and for small and large values of the SINR decoding threshold. By optimizing the system with the aid of the proposed analytical frameworks, it is shown that the proposed TIN-based scheduling algorithm outperforms conventional cellular networks in terms of effective coverage probability and effective average rate.

Interesting future works include the analysis of the trade-off between average rate and user fairness, as well as the energy efficiency gain provided by the proposed scheduling algorithm.
\appendices
\section{Proof of Theorem 1}
\label{proof_thm1}
The SINR coverage probability of the typical UE given that it satisfies the TIN optimality conditions can be expressed as follows:
\begin{equation}
\begin{split}
\mathcal{C} &= \int_0^\infty \mathbb{P}\left[\mathrm{SINR}\geq\Theta\right]f_{X_{11}}\left(x_{11}|\mathbb{A_{UE}}\right)\mathrm{d}x_{11} \\
& \overset{(a)}=\frac{2\pi\lambda_{b}}{\mathbb{P\left[A_{UE}\right]}}\int_0^\infty \mathbb{P}\left[\frac{h_{\mathrm{11}}x_{\mathrm{11}}^{-\alpha}}{I+\frac{N}{P}}\geq\Theta\right]x_{11}
\mathrm{e}^{-\pi\lambda_{b}\max^{2}\left(x_{11},x_{11}^{\mu/2}\left(\frac{P}{N}\right)^{\frac{2-\mu}{2\alpha}}\left(\frac{1}{M}\right)^{\frac{1}{2\alpha}}\right)}\mathrm{d}x_{11} \\
&\overset{(b)}=\frac{2\pi\lambda_{b}}{\mathbb{P\left[A_{UE}\right]}}\int_0^\infty \mathrm{e}^{-x_{11}^{\alpha}\Theta \frac{N}{P}} \mathbb{E}_I \left[\mathrm{e}^{-x_{11}^\alpha\Theta I}\right] x_{11}\mathrm{e}^{-\pi\lambda_{b}\max^{2}\left(x_{11},x_{11}^{\frac{\mu}{2}}\left(\frac{P}{N}\right)^{\frac{2-\mu}{2\alpha}}\left(\frac{1}{M}\right)^{\frac{1}{2\alpha}}\right)}\mathrm{d}x_{11},
\end{split}
\end{equation}
where $I$ denotes the other-cell interference, $\left(a\right)$ follows by using $f_{X_{11}}\left(x_{11}|\mathbb{A_{UE}}\right)$ in \eqref{pdf_X11}, $\left(b\right)$ follows because $h_{11}\sim \exp\left(1\right)$ is an exponential random variable, and $\mathbb{E}_I \left[\mathrm{e}^{-x_{11}^\alpha\Theta I}\right]=\mathcal{L}_I\left(x_{11}^\alpha\Theta\right)$ is the Laplace transform of the other-cell interference $I$.

Let us define $s=x_{11}^\alpha\Theta$. The Laplace transform $\mathcal{L}_I\left(s\right)$ can be written as follows:
\begin{equation}
\begin{split}
\mathcal{L}_I\left(s\right)&=\mathbb{E}_I \left[\mathrm{e}^{-sI}\right] =\mathbb{E}_{{h_i}, D_i} \left[\mathrm{e}^{-s\sum_{i\in\Phi_{b}^{'}}h_{i}D_{i}^{-\alpha}}\right]
\\ & =\mathbb{E}_{D_i} \prod_{i\in\Phi_{b}^{'}} \mathbb{E}_{h_i} \left[\mathrm{e}^{-s h_i D_i^{-\alpha}}\right]
\overset{(c)}= \mathbb{E}_{D_i} \prod_{i\in\Phi_{b}^{'}} \frac{1}{1+sD_i^{-\alpha}}
\\ & \overset{(d)}=\exp\left(-2\pi\lambda_b \mathbb{P\left[A_{UE}\right]} \int_{\max\left(x_{11},x_{11}^{\mu/2}\left(\frac{P}{N}\right)^{\frac{2-\mu}{2\alpha}}\left(\frac{1}{M}\right)^{\frac{1}{2\alpha}}\right)}^\infty \left(1-\frac{1}{1+su^{-\alpha}}\right)u\mathrm{d}u\right)\\
& =\exp\left(-\pi\lambda_b \mathbb{P\left[A_{UE}\right]} s^{\frac{2}{\alpha}} \int_{s^{\frac{-2}{\alpha}}\max^2\left(x_{11},x_{11}^{\mu/2}\left(\frac{P}{N}\right)^{\frac{2-\mu}{2\alpha}}\left(\frac{1}{M}\right)^{\frac{1}{2\alpha}}\right)}^\infty \frac{1}{1+z^{\alpha/2}}\mathrm{d}z\right),
\label{laplace1_sinr}
\end{split}
\end{equation}
where $\left(c\right)$ follows by computing the expectation with respect to $h_i$, and $\left(d\right)$ follows from the probability generating functional theorem of PPPs \cite{Geff_DL} by assuming that the point process of interfering BSs is an inhomogeneous PPP whose density is given in \eqref{den_interf_BS}. The proof follows with the aid of simple algebraic manipulations.
\ifCLASSOPTIONcaptionsoff
  \newpage
\fi
\renewcommand\refname{References} 
\bibliographystyle{IEEEtran}
\bibliography{IEEEabrv,Reference}


%
%
%
%
%
\end{document}